\documentclass{amsart}
\usepackage{enumitem}
\usepackage[backend=bibtex,style=numeric]{biblatex}
\usepackage{amsmath, amsthm, amssymb, amsfonts}
\usepackage{mystyle}
\usepackage{tikz-cd}
\usepackage{tikz}
\usepackage{float}
\usepackage{appendix}
\usepackage{hyperref}
\usepackage[utf8]{inputenc}
\usetikzlibrary{calc}

\title{Heat Kernel Renormalization on Manifolds with Boundary}
\author{Benjamin I. Albert}

\begin{document}
\maketitle
\begin{abstract}
  In the monograph \emph{Renormalization and Effective Field Theory}, Costello made two major advances towards the mathematical formulation of quantum field theory.  Firstly, he developed an inductive position space renormalization procedure for constructing effective field theories that is based on heat kernel regularization of the propagator.  Secondly, he gave a rigorous formulation of quantum gauge theory within effective field theory that makes use of the BV formalism.  In this work, we extend Costello's inductive renormalization procedure from manifolds without boundary to a class of manifolds with boundary.  In addition, we reorganize the presentation of the preexisting material, filling in details and strengthening the results.  

\end{abstract}

\tableofcontents{}
\section{Introduction}

Effective field theory and its renormalization group were developed by Wilson and by Kadanoff separately in papers in the late 1960s and early 1970s (see \cite{Wilson_1975} for a contemporaneous review).  There are many variants, but the basic theme involves two steps: mode elimination and rescaling \cite{Kopietz_2010} \cite{Costello_2011}.  In this introduction, we shall focus on the intuitive idea of mode elimination in order to motivate the body of the paper.

For simplicity, we will describe mode elimination for a scalar field theory on a compact manifold $M$.  Suppose that we have an action functional $S[\Lambda](\phi)$ describing physics below an energy scale $\Lambda$.  That is, $S[\Lambda](\phi)$ is a functional of $\phi \in \mathcal{E}_{[0, \Lambda]}$, where $\mathcal{E}_{[0, \Lambda]}$ is the span of the scalar fields with energy eigenvalues in the interval $[0, \Lambda]$.  Given energy scales $\Lambda_L < \Lambda_H$, the action functional $S[\Lambda_L](\phi)$ describing physics at scale $\Lambda_L$ should be given by ``eliminating the modes'' of $S[\Lambda_H](\phi)$ with energy between $\Lambda_L$ and scale $\Lambda_H$.  This is described by the renormalization group equation (RGE)
\begin{align}
  e^{S[\Lambda_L](\phi)/\hbar} = \int_{\phi' \in \mathcal{E}_{(\Lambda_L, \Lambda_H]}} e^{S[\Lambda_H](\phi + \phi')/\hbar} \mathcal{D} \phi'.
\end{align}
where the integral is over $\mathcal{E}_{(\Lambda_L, \Lambda_H]}$, the span of the fields with energy in the interval $(\Lambda_L, \Lambda_H]$.  Equivalently, we can write 
\begin{align}
  S[\Lambda_L](\phi) = \hbar \log \int_{\phi' \in \mathcal{E}_{(\Lambda_L, \Lambda_H]}} e^{S[\Lambda_H](\phi + \phi')/\hbar} \mathcal{D} \phi'.
\end{align}
In order to define the effective action $S[\Lambda](\phi)$, one might be tempted to let the higher energy cutoff go to infinity and write
\begin{align}\label{eq:s_lambda}
  S[\Lambda](\phi) = \lim_{\Lambda' \to \infty}\hbar \log \int_{\phi' \in \mathcal{E}_{(\Lambda, \Lambda']}} e^{S(\phi + \phi')/\hbar} \mathcal{D} \phi'.
\end{align}
where $S(\phi)$ is the classical action, but this limit will almost always not exist due to ultraviolet divergences.  However, the limit should exist after an appropriate renormalization of the functional integral in the right-hand side of (\ref{eq:s_lambda}) by adding counterterms depending on $\hbar$ to the classical action $S(\phi)$.

It is easier to work with the renormalization group equation when it is expressed in terms of the interaction part of the action functional.  Let $D$ be the Laplacian acting on smooth functions on $M$, a compact manifold as before.  Assume that the action is of the form
\begin{align}
  S(\phi) = -\frac{1}{2}\langle \phi, D  \phi \rangle + I(\phi).
\end{align}
where $\langle \phi, D\phi \rangle = \int_M \phi D \phi$ is the quadratic part of the action and $I(\phi)$ has cubic and higher degree parts. 
Because $\phi \in \mathcal{E}_{[0, \Lambda]}$ and $\phi' \in \mathcal{E}_{(\Lambda, \infty)}$ are orthogonal, we have
\begin{align}
  S(\phi + \phi') = -\frac{1}{2}\langle \phi, D\phi \rangle - \frac{1}{2} \langle \phi', D  \phi' \rangle + I(\phi + \phi').
\end{align}
Letting $S[\Lambda](\phi) = -\frac{1}{2}\langle \phi, D \phi \rangle + I[\Lambda](\phi)$, the renormalization group equation simpifies to 
\begin{align}\label{eq:rge_int}
  e^{I[\Lambda_L](\phi)/\hbar} = \int_{\phi' \in \mathcal{E}_{(\Lambda_L, \Lambda_H]}} e^{-\frac{1}{2}\langle \phi', D \phi' \rangle/\hbar + I[\Lambda_H](\phi + \phi')/\hbar} \mathcal{D} \phi'.
\end{align}
or equivalently
\begin{align}\label{eq:rge_int_log}
  I[\Lambda_L](\phi) = \hbar \log \int_{\phi' \in \mathcal{E}_{(\Lambda_L, \Lambda_H]}} e^{-\frac{1}{2}\langle \phi', D \phi' \rangle/\hbar + I[\Lambda_H](\phi + \phi')/\hbar} \mathcal{D} \phi'.
\end{align}

Let $P = P(\Lambda_L, \Lambda_H)$ be the inverse of the quadratic form $\langle \phi', D \phi' \rangle$ on $\mathcal{E}_{(\Lambda_L, \Lambda_H]}$ and let $\partial_P$ be the operator that contracts functionals of the fields in $\mathcal{E}_{(\Lambda_L, \Lambda_H]}$ with $P$.  Using Wick's theorem on the finite dimensional vector space $\mathcal{E}_{(\Lambda_L, \Lambda_H]}$ to rewrite the right-hand side of \eqref{eq:rge_int}, we have
\begin{align}
  \exp(I[\Lambda_L]/\hbar) = \exp(\hbar \partial_{P(\Lambda_L, \Lambda_H)})\exp(I[\Lambda_H]/\hbar)
\end{align}
and thus we can also rewrite \eqref{eq:rge_int_log} as
\begin{align}
  I[\Lambda_L] = \hbar \log[\exp(\hbar \partial_{P(\Lambda_L, \Lambda_H)})\exp(I[\Lambda_H]/\hbar)].
\end{align}

While the version of effective field theory with sharp energy cutoffs described above paints an intuitive physical picture, there are disadvantages to working with it, as discussed in \cite{Costello_2011}.  In particular, there does not seem to be a systematic way to construct counterterms and in so doing also ensure in quantum gauge theory that the obstruction to the existence of quantum gauge symmetry is a local functional.

Costello gives an alternate approach that comes from the relationship between the propagator without cutoff $P(0, \infty)$ (Green's function) and the heat kernel.  Let $K_t(x, y)$ be the heat kernel for $D$.  That is 
\begin{align}
  \partial_t K_t(x, y) + D_xK_t(x, y) = 0
\end{align}
and $\lim_{t \to 0^+} \int_M K_t(x, y) \phi(y)\, dy = \phi(x)$.  Then
\begin{align}
  P(0, \infty) = G(x, y) = \int_0^\infty K_t \, dt
\end{align}
is an inverse to $D$ on $\mathcal{E}_{(0, \infty)}$; that is, away from the energy zero fields.

Instead of cutting off the space of fields, we work with the entire space of fields $\mathcal{E}$ and introduce the regularized propagator
\begin{align}\label{eq:req_prop}
  P_\epsilon^L = \int_\epsilon^L K_t \, dt
\end{align}
An effective field theory now becomes a collection of ``length'' (or ``time'') scale regularized interactions satisfying
\begin{align}
  I[L] =  \hbar \log \left[\exp\left(\hbar \partial_{P_\epsilon^L}\right)\exp\left(I[\epsilon] /\hbar\right)\right].
\end{align}

We now naively might try to define the scale $L$ effective interaction as
\begin{align}\label{eq:naivequant}
  I[L] = \lim_{\epsilon \to 0^+}\hbar \log \left[\exp\left(\hbar \partial_{P_\epsilon^L}\right)\exp\left(I/\hbar\right)\right]
\end{align}
However, this limit may not exist and expression then has to be renormalized.  That is, an $\epsilon$ dependent interaction functional $I(\epsilon)$ of counterterms for $I$ is constructed such that
\begin{align}
  I[L] = \lim_{\epsilon \to 0^+} \hbar \log \left[\exp\left(\hbar \partial_{P_\epsilon^L}\right)\exp\left((I - I(\epsilon))/\hbar\right)\right]
\end{align}
exists.

In Section \ref{sec:fey_diag_exp}, we define the spaces to which the propagator $P$ and the interaction functional $I$ belong.  We define stable Feynman graphs which give a way of organizing the combinatorics of the contractions in
\begin{align}
  V(P, I) := \exp\left(\hbar \partial_P\right)\exp\left(I /\hbar\right)
\end{align}
and $W(P, I) = \hbar \log V(P, I)$.
Theorem \ref{thm:feynm-diagr-expans} expresses $V(P, I)$ as a summation over all stable graphs while Corollary \ref{cor:feynm-diagr-expans} expresses $W(P, I)$ as a summation over connected stable graphs.

In Section \ref{sec:wick_thm}, we state and prove several variants of Wick's theorem.  In \ref{sec:wick_ab}, we calculate the $1$ dimensional Gaussian integral
\begin{align}
  I_{m, \alpha}(a, b) = \int_a^b x^m e^{-\alpha x^2/2} \, dx
\end{align}
in terms of $I_{0, \alpha}(a, b)$ and $J_{i, \alpha}(a, b) = x^{i}e^{-\alpha x^2/2}|_{x = a}^{x = b}$ for $i < m$.  The formula reduces to expected results on $\RR$ and $\RR_{\geq 0}$ which are recalled in \ref{sec:wick_r} and \ref{sec:wick_rplus} respectively.  In \ref{sec:gener-wicks-theor}, we generalize the formula for $I_{m, \alpha}(a, b)$ to a formula for
\begin{align}
  I_{m, \alpha, \beta}(a, b) = \int_a^b x^m e^{-\alpha x^2/2 + \beta x} \, dx.
\end{align}
The proof, which is analogous to the one for the case $\beta = 0$ in \ref{sec:wick_ab} is omitted.
The next two sections are focused on the many variables Wick's theorem.  That is, the computation of the integral 
\begin{align}\label{eq:int_PetoQ}
  \int_Px_{m_1} \dots x_{m_k} e^{-Q(x)/2} \, dx
\end{align}
where $Q(x)$ is a nondegenerate quadratic form and $P$ is a polyhedron.
In \ref{sec:gau_sev}, we recall the standard statement of Wick's theorem on $P = \RR^n$ and give a proof by diagonalizing the quadratic form and applying the result of \ref{sec:wick_r}.  This can be used to calculate the counterterms on $\RR^n$ in \ref{sec:count_rn}.  In \ref{sec:wicks-theor-polyt} it is shown that the result of \ref{sec:wick_ab} is sufficient to compute (\ref{eq:int_PetoQ}) inductively, when $P$ is any polyhedron, although the answer cannot be simplified much.
 
Section \ref{sec:heat_kern_count}, forms the body of the paper.  We begin with \ref{sec:some-examples}, where the construction of the counterterms for a general Feynman weight is motivated by working out the example of a Feynman weight associated to a particular $1$-loop graph in the $\phi^4$-theory.  The renormalization procedure is based on the ability to cover $(0, \infty)^k$ and a fortiori $(\epsilon, L)^k$ by sets defined by inequalities of the form $t_i \leq t_j^R$, where $R > 1$.  In \ref{sec:covering_egamma}, the covering lemma that was proved by Costello in \cite{Costello_2011} is strengthened and proved.  Also, much more detail about the nature of the sets in the cover is given.  Other preliminary concepts needed for the renormalization procedure like local functionals and the form of Feynman weights associated to them are then discussed in \ref{sec:local-funct-feynm}.

In \ref{sec:count_rn}, we formulate Costello's renormalization procedure on $\RR^n$.  We give explicit formulas whenever possible and fill in a few steps in the argument omitted by Costello, such as the introduction of what we call spanning tree coordinates.  The basic construction that is carried out in \ref{sec:count_rn} is the ingredient that is then used in \ref{sec:count-rrn:-error} inductively to provide counterterms on each of the sets in the cover of $(\epsilon, L)^k$, where $k$ is the number of edges of the Feynman graph whose weight we are trying to renormalize.  In \ref{sec:count-rrn:-error}, we also importantly treat how to control the error of the construction of \ref{sec:count_rn}, which is an essential part of the inductive procedure.

In \ref{sec:count_upp_half}, the renormalization procedure is adapted to $\HH^n$, the upper half space with the Euclidean metric.  The procedure does not carry over without modification since the quadratic form in the integral computing the Feynman weight now depends on the center of mass coordinate.  Luckily, this difficulty can be circumvented when necessary by doing an additional Taylor expansion in the normal to the boundary component of the center of mass coordinate.  The counterterms have a more complicated form than those on $\RR^n$, but we argue that the inductive procedure of \ref{sec:count-rrn:-error} can be carried out with appropriate modifications.

In \ref{sec:count-comp-manif}, we discuss the counterterm construction on a Riemannian manifold without boundary following similar reasoning to \cite{Costello_2011}, but with modifications.  We correct what seems to be an oversight in Costello's reasoning.  In particular, on a Riemannian manifold $M$, Costello uses the asymptotic expansion of the heat kernel $K_t(x, y) \sim e^{-d(x, y)^2/4t}\sum_i \phi_i(x, y)t^i$, but for each chart in a cover replaces the geodesic distance $d(x, y)$ with the coordinate distance $\|x - y\|$.  Thus, taking a partition of unity, the Feynman weight under consideration becomes a sum of integrals whose integrands will contain the exponential of a quadratic form, which allows us to apply Wick's theorem to construct counterterms.  However, it is not correct that $K_t(x, y) \sim e^{-\|x - y\|^2/4t}\sum_i \phi_i(x, y)t^i$, at least not uniformly in $x$ and $y$.  We show how this issue can be resolved by adding an additional step to the procedure.  In the end, while structure of the counterterms will be more complicated than on $\RR^n$, one can still bound the error.  The inductive step in the constuction thus remains valid.

The culmination of these results is \ref{sec:count-comp-manif-bdry} where we show the renormalization procedure can be carried out on a Riemannian manifold with boundary with a cylindrical collar neighborhood, with the argument reducing to that of \ref{sec:count_upp_half} near the boundary and that of \ref{sec:count-comp-manif} away from the boundary.  We also comment on generalizing the discussion to Riemannian manifolds with boundary for which the double (always a smooth manifold) is also a Riemannian manifold (the theory of such manifolds is described in detail in Appendix \ref{sec:riemdouble})

In Section \ref{sec:constr_eft}, we move beyond the construction of counterterms for each Feynman weight and construct the counterterms $I^{CT}(\epsilon)$ for the entire effective interaction, essentially following Costello's book \cite{Costello_2011} and an earlier paper \cite{2007arXiv0706.1533C}.

\section{Feynman Diagrams}
\label{sec:fey_diag_exp}
\subsection{General Setup}
Let $\mathcal{E}$ be a graded object in an appropriate symmetric monoidal category with dual objects, which contains a field $\KK = \RR, \CC$ as its monoidal unit.  For toy examples one can work with the category of finite dimensional vector spaces over $\KK$.  For quantum field theory one will need to work with an appropriate subcategory of the category of topological vector spaces where the identifications $(\mathcal{E} \otimes \mathcal{F})^* \cong \mathcal{E}^* \otimes \mathcal{F}^*$ and $\Hom(\mathcal{E}, \mathcal{F}) \cong \mathcal{E}^* \otimes \mathcal{F}$ can be made.  We will not dwell on the issue any further and direct the interested reader to the appendices of \cite{Costello_2011} and \cite{Costello_2016}.

Fix an element $P \in \Sym^2(\mathcal{E})$ which will be called a \emph{propagator}.  We define the algebra of formal power series on $\mathcal{E}$,
\begin{align}
  \mathcal{O}(\mathcal{E}) = \prod_{n \geq 0} \Hom(\otimes^n\mathcal{E}, \KK)_{S_n} = \prod_{n \geq 0} \Sym^n(\mathcal{E}^*)
\end{align}
Here $\Sym$ means taking coinvariants of the $n$-fold tensor product with respect to the symmetric group action.  An element of $I \in \mathcal{O}(\mathcal{E})[[\hbar]]$ is of the form $I = \sum_{i, k \geq 0} I_{i, k}\hbar^i$, where $I_{i, k} \in \Sym^k(\mathcal{E}^{*})$.  Let
\begin{align}
  \mathcal{O}(\mathcal{E})^{+}[[\hbar]] \subset \mathcal{O}(\mathcal{E})[[\hbar]]
\end{align}
be the functionals of the form $I = \sum_{i, k \geq 0} I_{i, k}\hbar^i$, where $I_{0, k} = 0$ for $k < 3$ and $I_{1, 0} = 0$.  We will see the reason for this restricted class of functionals later in the section.

We are interested in combinatorial formulas for ``functional integrals'' of the form
\begin{align}
  V(P, I) = e^{\hbar \partial_P}e^{I/\hbar}
\end{align}
and
\begin{align}
  W(P, I) = \hbar \log (e^{\hbar \partial_P}e^{I/\hbar}),
\end{align}
where $\partial_P$ denotes the contraction operator $\frac{1}{2}\sum_i \partial_{P^{(1)}_i} \partial_{P^{(2)}_i}$ for $P = \sum_i P_i^{(1)} \otimes P_i^{(2)}$.
Such formulas can be obtained directly and transparently by replacing the exponentials with their Taylor series and identifying the coefficients of the resulting formal series:
\begin{lem}[Pre-Feynman Expansion]\label{lem:pre-feyn}
  \begin{align}
    V(P, I) = \sum_{\{n_{i, k}\}}\sum_{j \geq 0} C(\{n_{i, k}\}, j) \hbar^{p(\{n_{i, k}\}, j)} \partial^j_P\prod_{i, k \geq 0}I_{i, k}^{n_{i, k}}
  \end{align}
  where in the outer summation, we sum over the collection of double sequences of almost all zero nonnegative integers $\{n_{i, k}\}_{i, k \geq 0}$ (double sequences of nonnegative integers $\{n_{i, k}\}_{i, k \geq 0}$ with the requirement that for all but finitely many $i, k$, $n_{i, k} = 0$).  We have defined
  \begin{align}
     C(\{n_{i, k}\}, j) = \frac{1}{j!}\prod_{i, k \geq 0}\frac{1}{n_{i, k}!}
  \end{align}
  and
  \begin{align}
    p(\{n_{i, k}\}, j) = \sum_{i, k \geq 0} i\,n_{i, k} - \sum_{i, k \geq 0} n_{i, k} + j.
  \end{align}

\end{lem}
\begin{proof}
  By the multinomial formula
  \begin{align}
    \exp\left(\sum_{i, k \geq 0} I_{i, k} \hbar^{i - 1}\right) &= \sum_{j \geq 0} \frac{(\sum_{i, k \geq 0} I_{i, k} \hbar^{i - 1})^j }{j!}\\
    &= \sum_{j \geq 0} \sum_{\,|\{n_{i, k}\}| = j \,} \prod_{i, k \geq 0} \frac{\hbar^{(i - 1)n_{i, k}}}{n_{i, k}!} I_{i, k}^{n_{i, k}},
  \end{align}
where the inner sum is over sequences of nonnegative numbers $\{n_{i, k}\}$ such that $\sum_{i, k \geq 0} n_{i, k} = j$.
We can reexpress this as a single sum over all double sequences of almost all zero nonnegative integers $\{n_{i, k}\}$
\begin{align}
  \exp\left(\sum_{i, k \geq 0} I_{i, k} \hbar^{i - 1}\right) = \sum_{\{n_{i, k}\}} \prod_{i, k \geq 0} \frac{\hbar^{(i - 1)n_{i, k}}}{n_{i, k}!} I_{i, k}^{n_{i, k}}.
\end{align}
Thus,
  \begin{align}
    V(P, I) &= \sum_{\{n_{i, k}\}}\left(\sum_{j \geq 0} \frac{\hbar^j}{j!} \partial^j_P\right)\prod_{i, k} \frac{\hbar^{(i - 1)n_{i, k}}}{n_{i, k}!} I_{i, k}^{n_{i, k}}\\
    &= \sum_{\{n_{i, k}\}}\sum_{j \geq 0} C(\{n_{i, k}\}, j) \hbar^{p(\{n_{i, k}\}, j)} \partial^j_P\prod_{i, k}I_{i, k}^{n_{i, k}}
  \end{align}
\end{proof}
It remains to investigate the combinatorial structure of the expression
\begin{align}
  \partial^j_P\prod_{i, k}I_{i, k}^{n_{i, k}}.
\end{align}
Before doing so, we shall make a definition.

\begin{defin}
  A \emph{stable graph} is defined by
  \begin{description}
  \item[$V(\gamma)$] a set of vertices
  \item[$E(\gamma)$] a set of edges each connecting two vertices
  \item[$T(\gamma)$] a set of tails each connected to one vertex
  \end{description}
  and a function $g: V(\gamma) \to \ZZ^{\geq 0}$ associating a ``genus'' to each vertex.
\end{defin}
A note for what follows in the next section: there is a natural partial ordering on the vertices of a stable graph: We write $v_1 \preceq v_2$ if $g(v_1) < g(v_2)$ or $g(v_1) = g(v_2)$ and $k_1 \leq k_2$ where $v_1$ has valency $k_1$ and $v_2$ has valency $k_2$.

\subsection{Feynman Diagram Expansion}
\label{sec:fey_ungraded}
Begin with the expression of Lemma \ref{lem:pre-feyn}:
\begin{align}
  V(P, I) = \sum_{\{n_{i, k}\}}\sum_{j \geq 0} \left(\frac{1}{j!2^j}\prod_{i, k \geq 0}\frac{1}{n_{i, k}!}\right) \hbar^{p(\{n_{i, k}\}, j)} \left(\sum_l \partial_{P_l^{(1)}}\partial_{P_l^{(2)}}\right)^j\prod_{i, k \geq 0}I_{i, k}^{n_{i, k}}.
\end{align}
Make the substitution $I_{i, k} = S^kI_{i, k}/k!$ where $S^kI_{i, k} = \sum_{\sigma \in S_k} I_{i, k}^\sigma = k! I_{i, k}$.  Then
\begin{align}
    V(P, I) = \sum_{\{n_{i, k}\}}\sum_{j \geq 0} \widetilde{C}(\{n_{i, k}\}, j) \hbar^{p(\{n_{i, k}\}, j)} \left(\sum_l \partial_{P_l^{(1)}}\partial_{P_l^{(2)}}\right)^j\prod_{i, k \geq 0}(S^kI_{i, k})^{n_{i, k}}.
\end{align}
where
\begin{align}
  \widetilde{C}(\{n_{i, k}\}, j) = \left(\frac{1}{j!2^j}\prod_{i, k}\frac{1}{n_{i, k}!(k!)^{n_{i, k}}}\right) 
\end{align}

Note that
\begin{align}
  \left(\sum_l \partial_{P_l^{(1)}}\partial_{P_l^{(2)}}\right)^j\prod_{i, k \geq 0}(S^kI_{i, k})^{n_{i, k}}.
\end{align}
 will be a sum over contractions that can be parametrized by injections $Q: H_j \to V(\{n_{i, k}\})$ of the set $H_j = \{1^{(1)}, 1^{(2)}, \dots, j^{(1)}, j^{(2)}\}$ into the set of inputs to the interactions
 \begin{align}
   V(\{n_{i, k}\}) = \bigsqcup_{i, k \geq 0}  (\{1^{(1)}, \dots, k^{(1)}\} \cup \dots \cup \{1^{(n_{i, k})}, \dots, k^{(n_{i, k})}\})
 \end{align}

Because $P$ is degree $0$ (in the grading on $\Sym^2 \mathcal{E}$) by assumption, we can reorder the contractions so that $Q(1^{(1)}), \dots, Q(j^{(1)})$ is in ascending order.  There are $j!$ contractions that will be reordered to the same contraction in this way.  Also, because $P$ is a symmetric tensor, we can also reorder so that $Q(1^{(1)})$ comes before $Q(1^{(2)})$, $Q(2^{(1)})$ comes before $Q(2^{(2)})$, and so forth.  There are $2^j$ contractions that will be reordered to the same contraction in this way.

Injections up to these reorderings are in one-to-one correspondence with partitions of $V(\{n_{i, k}\})$ into $j$ subsets with two elements and $1$ additional subset containing the remaining $|V(\{n_{i, k}\})| - 2j$ elements.  Let $\mathcal{Q}(\{n_{i, k}\}, j)$ be the collection of such partitions and for $Q \in \mathcal{Q}(\{n_{i, k}\}, j)$ let $w_Q(P, I)$ denote the corresponding contraction.

Then
\begin{align}\label{eq:v_pi_1}
  V(P, I) = \sum_{\{n_{i, k}\}}\sum_{j \geq 0} \sum_{Q \in \mathcal{Q}(\{n_{i, k}\}, j)}\left(\prod_{i, k \geq 0}\frac{1}{n_{i, k}!(k!)^{n_{i, k}}}\right) \hbar^{p(\{n_{i, k}\}, j)}w_Q(P, I)
\end{align}

Any partition $Q \in \mathcal{Q}(\{n_{i, k}\}, j)$ determines a stable graph $\gamma$ in an obvious way.  Consider $\mathcal{Q}_\gamma(\{n_{i, k}\}, j)$, the collection of partitions which determine the same stable graph $\gamma$ (up to isomorphism).  Let $G(\{n_{i, k}\}) = \prod_{i, k \geq 0}(S^{n_{i, k}}_{k} \rtimes S_{n_{i, k}})$. 
This acts on $V(\{n_{i, k}\})$ by permuting the $n_{i, k}$ interactions of type $i, k$ and their $k$ inputs.  As a consequence, it acts on $\mathcal{Q}(\{n_{i, k}\}, j)$.  In fact, it acts transitively on $\mathcal{Q}_\gamma(\{n_{i, k}\}, j)$.  The stabilizer subgroup of a given partition $Q \in \mathcal{Q}_\gamma(\{n_{i, k}\}, j)$ is equal to $\Aut(\gamma)$, the group of automorphisms of the stable graph $\gamma$ associated to $Q$.  By the orbit-stabilizer theorem, $|\mathcal{Q}_\gamma(\{n_{i, k}\}, j)|$, the number of partitions which determine the same stable graph $\gamma$ is given by
\begin{align}\label{eq:G_nik}
  \frac{|G(\{n_{i, k}\})|}{|\Aut(\gamma)|} =\frac{\prod_{i, k \geq 0}n_{i, k}!(k!)^{n_{i, k}}}{|\Aut(\gamma)|}
\end{align}

Therefore,
\begin{thm}[Feynman Diagram Expansion]\label{thm:feynm-diagr-expans}
  For a stable graph $\gamma$, we define
\begin{align}
  g(\gamma) = b(\gamma) + \sum_{v \in V(\gamma)} g(v) 
\end{align}
where $b(\gamma)$ is the first Betti number of $\gamma$, the number of independent loops.  Let $C(\gamma)$ be the number of connected components of $\gamma$.  Then
  \begin{align}
  V(P, I) = \sum_\gamma\frac{1}{|\Aut(\gamma)|}\hbar^{g(\gamma) - C(\gamma)} w_\gamma(P, I)
\end{align}
with the sum being over all (isomorphism classes of) stable graphs.
\end{thm}
\begin{proof}
  The constant $p(\{n_{i, k}\}, j) = \sum_{i, k} in_{i, k} - \sum_{i, k} n_{i, k} + j$ has a very simple interpretation in terms of the stable graph $\gamma$ since
  \begin{align}
    \sum_{v \in V(\gamma)} g(v) &= \sum_{i, k} i\,n_{i, k},
  \end{align}
$|V(\gamma)| = \sum_{i, k} n_{i, k}$ and $|E(\gamma)| = j$.    Using the fact that
\begin{align}
  b(\gamma) = |E(\gamma)| - |V(\gamma)| + C(\gamma),
\end{align}
and the definition
\begin{align}
  g(\gamma) = b(\gamma) + \sum_{v \in V(\gamma)} g(v)
\end{align}
we have
\begin{align}
  g(\gamma) - C(\gamma) = p(\{n_{i, k}\}, j).
\end{align}
Lastly define $w_\gamma(P, I)$ to be $w_Q(P, I)$ where $Q$ is a partition of $V(\{n_{i, k}\})$ that determines $\gamma$.  The formula now follows from \eqref{eq:v_pi_1} and \eqref{eq:G_nik}.
\end{proof}

Now we describe a combinatorial formula for $W(P, I) = \hbar \log(e^{\hbar \partial_P} e^{I/\hbar})$:
\begin{cor}\label{cor:feynm-diagr-expans}
  \begin{align}
    W(P, I) = \sum_{\gamma\text{ conn}}\frac{1}{|\Aut(\gamma)|}\hbar^{g(\gamma)} w_\gamma(P, I)
  \end{align}
\end{cor}
\begin{proof}
  If $\gamma_1 \sqcup \dots \sqcup \gamma_k$ is the disjoint union of not necessarily distinct (isomorphism classes of) connected stable graphs $\gamma_1, \dots, \gamma_k$, then it is clear that
  \begin{align}
    g(\gamma_1 \sqcup \dots \sqcup \gamma_k) &= g(\gamma_1) + \dots +  g(\gamma_k)\\
    C(\gamma_1 \sqcup \dots \sqcup \gamma_k) &= C(\gamma_1) + \dots + C(\gamma_2)
  \end{align}
  and if $\gamma = (\sqcup^{k_1} \gamma_1) \sqcup \dots \sqcup(\sqcup^{k_n} \gamma_n)$ where $\gamma_1, \dots, \gamma_n$ are distinct
  \begin{align}
    |\Aut(\gamma)| = k_1!\dots k_n!|\Aut(\gamma_1)|^{k_1}\dots |\Aut(\gamma_n)|^{k_n}
  \end{align}
Thus,
  \begin{align}
    \exp\left( W(P, I)/ \hbar\right) &= \exp\left(\sum_{\gamma\text{ conn}}\frac{1}{|\Aut(\gamma)|}\hbar^{g(\gamma) - 1} w_\gamma(P, I)\right)\\
                                     &= \sum_{\{k_\gamma\}} \prod_{\gamma \text{ conn}} \frac{1}{|\Aut(\gamma)|^{k_\gamma} k_\gamma!}\hbar^{k_\gamma(g(\gamma) - 1)} w_{\sqcup_{k_\gamma}\gamma}(P, I)\\
    &= \sum_\gamma\frac{1}{|\Aut(\gamma)|}\hbar^{g(\gamma) - C(\gamma)} w_\gamma(P, I)
  \end{align}
  In the second line above, the outer summation is over the collection of all labelings of connected stable graphs by non-negative integers $\{k_\gamma\}$, where $k_\gamma = 0$ for all but finitely many $\gamma$.
\end{proof}
\begin{cor}
  For $I \in \mathcal{O}(\mathcal{E})^{+}[[\hbar]]$,
  \begin{align}
    W(P, I) \in \mathcal{O}(\mathcal{E})^{+}[[\hbar]]
  \end{align}
\end{cor}

\section{Wick's Theorem}
\label{sec:wick_thm}

\subsection{Wick's Theorem on \texorpdfstring{$\RR$}{R}}
\label{sec:wick_r}

In one variable, Wick's theorem is the statement
\begin{align}
  \int_{-\infty}^\infty x^m e^{-\alpha x^2/2} \, dx  &=  C_m\frac{\sqrt{2 \pi}}{\alpha^{(m + 1)/2}}
\end{align}
where
\begin{align}
  C_m = 
  \begin{cases}
    (m - 1)!!  & \text{if $m$ is even}\\
    0 & \text{if $m$ is odd}.
  \end{cases}
\end{align}
Above, $(m - 1)!! = (m - 1)(m - 3) \dots 1$ denotes the double factorial.

\subsection{Wick's theorem on \texorpdfstring{$\RR_{\geq 0}$}{R+}}
\label{sec:wick_rplus}
The following will be seen as a corollary of Wick's theorem on an interval $(a, b)$ as stated in Section \ref{sec:wick_ab}:
\begin{align}
  \int_{0}^\infty x^m e^{-\alpha x^2/2} \, dx &= C_m\frac{\sqrt{2\pi}}{2}\frac{1}{\alpha^{(m + 1)/2}} + \tilde{C}_m\frac{1}{\alpha^{(m + 1)/2}} 
\end{align}
where $C_m$ is as above and
\begin{align}
  \tilde{C}_m =
  \begin{cases}
    0 & \text{$m$ even} \\
    (m - 1)!! & \text{$m$ odd}
  \end{cases}
\end{align}

\subsection{Generalized Wick's Theorem on $\RR$}
\label{sec:gen-wicks-on-r}
The following can be seen as consequence of the generalized Wick's theorem on an interval $(a, b)$.  It can also be seen as consequence of the Wick's theorem on $\RR$ by completing the square in the exponent and changing variables:
\begin{align}
  \int_{-\infty}^\infty x^m e^{-\alpha x^2/2 + \beta x} \, dx  &= e^{-\beta^2/2\alpha}\int_{-\infty}^\infty  (x + \beta/\alpha)^m e^{-\alpha x^2/2} \, dx \\
                                                               &= e^{-\beta^2/2\alpha} \sum_{i = 0}^m \binom{m}{i} \frac{C_i \sqrt{2\pi}}{\alpha^{(i + 1)/2}} \frac{\beta^{m - i}}{\alpha^{m - i}}.
\end{align}

\subsection{Wick's Theorem on \texorpdfstring{$(a, b)$}{(a,b)}}
\label{sec:wick_ab}
There are several ways of proving the formula for $\RR$ which one might try to adapt, such as integration by parts and differentiation of $\alpha$ under the integral sign.  The proof by integration by parts is the one that works here.

We wish to compute the integral
\begin{align}
  I_{m, \alpha}(a, b) = \int_a^b x^m e^{-\alpha x^2/2} \, dx
\end{align}
for $-\infty \leq a \leq b \leq\infty$ and to check that the result agrees with the standard formula for $a = -\infty$ and $b = \infty$.  Let
\begin{align}
  J_{m, \alpha}(a, b) = x^me^{-\alpha x^2/2}\bigg|_{x = a}^{x = b}.
\end{align}
By integration by parts,
\begin{align}\label{eq:int_by_parts}
  I_{m, \alpha}(a, b) = \frac{m - 1}{\alpha}I_{m - 2, \alpha}(a, b) - \frac{1}{\alpha} J_{m - 1, \alpha}(a, b).
\end{align}
For $m$ even, we can thus express $I_{m, \alpha}(a, b)$ in terms of $I_{0, \alpha}(a, b)$ and $J_{l, \alpha}(a, b)$ where $l$ ranges over odd integers less than $m$.
For $m$ odd, since $I_{1, \alpha}(a, b) = -(1/\alpha)J_{0, \alpha}(a, b)$, we can express $I_{m, \alpha}(a, b)$ in terms of $J_{l, \alpha}(a, b)$, where $l$ ranges over even integers less than $m$.

We can prove a precise formula by induction:
\begin{prop}\label{prop:I_malpha}
  \begin{align}
    I_{m, \alpha}(a, b) = \frac{C_{m}}{\alpha^{m/2}} I_{0, \alpha}(a, b) - \sum_{i = 0}^{\lfloor\frac{m - 1}{2}\rfloor} \frac{\tilde{C}_{i, m}}{\alpha^{i + 1}} J_{m - 1 - 2i, \alpha}(a, b),
  \end{align}
where $C_m = (m - 1)!!$ when $m$ is even and $C_m = 0$ when $m$ is odd and for all $m$
\begin{align}
  \tilde C_{i, m} = \frac{(m - 1)!!}{(m - 1 - 2i)!!}
\end{align}

\end{prop}

\begin{proof}
  The even and odd base cases when $m = 0$ and $m = 1$ are clearly satisfied.  Suppose the result is true for $I_{m, \alpha}(a, b)$.  Then using \eqref{eq:int_by_parts},
  \begin{align}
    I_{m + 2, \alpha}(a, b) &= \frac{m + 1}{\alpha}\frac{C_m}{\alpha^{m/2}} - \frac{m + 1}{\alpha} \sum_{i = 0}^{\lfloor \frac{m - 1}{2} \rfloor} \frac{\tilde C_{i, m}}{\alpha^{i + 1}} J_{m - 1 - 2i, \alpha}(a, b)\\
    &- \frac{1}{\alpha}J_{m + 1, \alpha}(a, b),
  \end{align}
  and
  \begin{align}
    (m + 1) \tilde C_{i, m} &= \frac{(m + 1)!!}{(m + 1 - 2(i + 1))!!}\\
    &= \tilde C_{i + 1, m + 2}
  \end{align}
  so
  \begin{align}
    &\frac{m + 1}{\alpha} \sum_{i = 0}^{\lfloor \frac{m - 1}{2} \rfloor} \frac{\tilde C_{i, m}}{\alpha^{i + 1}} J_{m - 1 - 2i, \alpha}(a, b) \\
    &= \sum_{i = 0}^{\lfloor \frac{m - 1}{2} \rfloor} \frac{\tilde C_{i + 1, m + 2}}{\alpha^{(i + 1) + 1}} J_{(m + 2) - 1 - 2(i + 1), \alpha}(a, b) \\
    &= \sum_{i = 1}^{\lfloor \frac{(m + 2) + 1}{2} \rfloor} \frac{\tilde C_{i, m + 2}}{\alpha^{i + 1}} J_{(m + 2) - 1 - 2i}(a, b)
  \end{align}

  The induction step is now completed by employing the fact that
  \begin{align}
    \frac{(m + 1)C_m}{\alpha \alpha^{m/2}} = \frac{C_{m + 2}}{\alpha^{(m + 2)/2}}
  \end{align}
\end{proof}
For $a = -\infty$ and $b = \infty$, we have $J_{l, \alpha}(a, b) = 0$ for all $l$ and $I_{0, \alpha}(a, b) = \sqrt{2\pi/\alpha}$, so we recover the statement of Wick's theorem on $\RR$.

Similarly, if $a = 0$ and $b = \infty$, then $J_{l, \alpha}(a, b) = 0$ for $l \neq 0$ and $J_{0, \alpha}(a, b) = -1$.  Since $I_{0, \alpha}(a, b) = \frac{1}{2}\sqrt{2\pi/\alpha}$ we recover the statement of Wick's theorem on $\RR_{\geq 0}$.

\subsection{Generalized Wick's Theorem on \texorpdfstring{$(a, b)$}{(a,b)}}
\label{sec:gener-wicks-theor}
In \ref{sec:wicks-theor-polyt} and later in \ref{sec:count_upp_half}, we shall encounter integrals over polyhedra of polynomials multiplied by exponentials of inhomogeneous quadratic forms.  Here we establish the one dimensional result that can be used iteratively to calculate such integrals explicitly.

We wish to compute the integral
\begin{align}
  I_{m, \alpha, \beta}(a, b) = \int_a^b x^m e^{-\alpha x^2/2 + \beta x} \, dx
\end{align}
for $-\infty \leq a \leq b \leq\infty$ and to check that the result agrees with the standard formula for $a = -\infty$ and $b = \infty$.  Let
\begin{align}
  J_{m, \alpha, \beta}(a, b) = x^me^{-\alpha x^2/2 + \beta x}\bigg|_{x = a}^{x = b}.
\end{align}
Firstly, \eqref{eq:int_by_parts} generalizes to 
\begin{align}
  I_{m, \alpha, \beta}(a, b) = - \frac{1}{\alpha} J_{m - 1, \alpha, \beta}(a, b) + \frac{\beta}{\alpha}I_{m - 1, \alpha, \beta}(a, b) + \frac{m - 1}{\alpha}I_{m - 2, \alpha, \beta}(a, b).
\end{align}
The following is a generalization of Proposition \ref{prop:I_malpha}
\begin{prop}
  \begin{align}
    I_{m, \alpha, \beta}(a, b) = &-\sum_{i = 0}^{m - 1}\sum_{\substack{\{r_j\}\\ \sum r_j = i}} \frac{\beta^{|r^{-1}(1)|} \prod_{k \in r^{-1}(2)}(m - 1 + s_k  - i)}{\alpha^{|l(r)| + 1}} J_{\alpha, \beta, m - i - 1}(a, b) \nonumber\\
    &+ \sum_{\substack{\{r_j\}\\ \sum r_j = m}} \frac{\beta^{|r^{-1}(1)|} \prod_{k \in r^{-1}(2)}(s_k - 1)}{\alpha^{|l(r)|}} I_{\alpha, \beta, 0}(a, b)
  \end{align}
where $\{r_j\}$ ranges over finite sequences such that $r_j \in \{1, 2\}$ for all $j$.  We use $l(r)$ to denote the length of the sequence $\{r_i\}$ and $s_k = \sum_{j = 1}^k r_j$.
\end{prop}
We shall not give the proof which is a straightforward induction like the proof of Proposition \ref{prop:I_malpha}.  However, let us just check that it reduces to the formula of Proposition \ref{prop:I_malpha} in the case that $\beta = 0$.  Since $0^0 = 1$ and $0^k = 0$ for $k > 0$ the only nonzero terms in the sums will come from sequences with $r^{-1}(1) = \emptyset$.  But there is exactly one such sequence such that $\sum r_j = m$ for $m$ even and it has $l(r) = m/2$ and no such sequences for $m$ odd.  It is clear that this then becomes the formula of Proposition \ref{prop:I_malpha}.

\subsection{Wick's theorem on \texorpdfstring{$\RR^n$}{R^n}}
\label{sec:gau_sev}

Suppose that $A$ is an invertible symmetric $n \times n$ matrix and consider the associated quadratic form $Q(x) = \langle x, Ax \rangle = x^iA_{ij}x^j$.  We wish to compute the integral
\begin{align}
  I_{J, A} &= \int_{\RR^n}x_{m_1} \dots x_{m_k} e^{-Q(x)/2} \, dx
\end{align}
where $J = (j_1, \dots, j_n)$ is a multi-index such that $x_{m_1}\dots x_{m_{k}} = x_1^{j_1}\dots x_n^{j_n}$.

\begin{thm}[Wick's Theorem on $\RR^n$]
  For $k$ even
  \begin{align}
    \int_{\RR^n}x_{m_1} \dots x_{m_k} e^{-Q(x)/2} \, dx =   \frac{\sqrt{2\pi}}{\sqrt{\det(A)}}\sum_{\beta}\prod_{j = 1}^{k/2}A^{-1}_{\beta_j^{(1)}, \beta_j^{(2)}}
  \end{align}
where the sum is over partitions of the multiset $\{m_1, \dots, m_k\}$ into $k/2$ subtuples of $2$ elements.  Here $\beta_j^{(1)}$ and $\beta_j^{(2)}$ denote respectively the first and second elements of the $j$-th multi-subset in the partition.
\end{thm}
\begin{proof}
  Let $D$ denote the diagonalization of $A$ and assume that $D$ has diagonal entries $\alpha_1, \dots, \alpha_n$.  In this new basis, using the change of basis matrix $S$, we have (in Einstein summation convention) a linear combination
  \begin{align}
    S^{i_1}_{m_1}\dots S^{i_k}_{m_k} \int_{\RR^n} y_{i_1} \dots y_{i_k} e^{-\alpha_1 y_1^2/2} \dots e^{-\alpha_n y_n^2/2}  \, dy.
  \end{align}
Apply Wick's theorem on $\RR$ separately in each variable.  Let $y_1^{l_1} \dots y_n^{l_n} = y_{i_1} \dots y_{i_k}$.  We have
\begin{align}
  \int_{\RR^n}y_1^{k_1} \dots y_n^{k_n} e^{-\alpha_1 y_1^2/2} \dots e^{-\alpha_n y_n^2/2}  \, dy &= \frac{1}{\sqrt{\alpha_1 \dots \alpha_n}} \prod_{i = 1}^n \frac{C_{k_i}}{\alpha_i^{k_i/2}} = \frac{1}{\sqrt{\det(A)}} \prod_{i = 1}^n \frac{C_{k_i}}{D_{ii}^{k_i/2}} \\                                                                                          &= \frac{(\sqrt{2\pi})^n}{\sqrt{\det(A)}}\sum_{\beta} \prod_{j = 1}^{k/2}D^{-1}_{\beta_j^{(1)}, \beta_j^{(2)}}
\end{align}
where the sum is over partitions of the multiset $\{i_1, \dots, i_k\}$ into $k/2$ multi-subsets of $2$ elements (there is only one such partition for which the product can be nonzero).  We then switch the order of summation so that the inner sum over partitions $\beta$ of $\{i_1, \dots, i_k\}$ becomes the outer sum over partitions $\beta$ of $\{m_1, \dots, m_k \}$; that is, 
\begin{align}
  \frac{(\sqrt{2\pi})^n}{\sqrt{\det(A)}}\sum_{i_1, \dots, i_k} \sum_{\beta} S^{i_1}_{m_1}\dots S^{i_k}_{m_k}\prod_{j = 1}^{k/2}D^{-1}_{\beta_j^{(1)}, \beta_j^{(2)}} &= \frac{(\sqrt{2\pi})^n}{\sqrt{\det(A)}} \sum_{\beta} \prod_{j = 1}^{k/2} S^{i_{2j - 1}}_{\beta_j^{(1)}}S^{i_{2j}}_{\beta_j^{(2)}}D^{-1}_{{i_{2j - 1}}, {i_{2j}}}\\
  &= \frac{(\sqrt{2\pi})^n}{\sqrt{\det(A)}}\sum_{\beta}\prod_{j = 1}^{k/2}A^{-1}_{\beta_j^{(1)}, \beta_j^{(2)}}
\end{align}
\end{proof}

\subsection{Generalized Wick's Theorem for Polyhedra}\label{sec:wicks-theor-polyt}
More generally, we wish to compute the integral
\begin{align}
  I_{J, A} &= \int_P x_{m_1} \dots x_{m_k} e^{-Q(x)/2} \, dx
\end{align}
where $J = (j_1, \dots, j_n)$ is a multi-index such that $x_1^{j_1}\dots x_n^{j_n} = x_{m_1}\dots x_{m_k}$ and $P$ is a polyhedron.  We can also assume that $Q(x) = \sum_{i,j}x_iA_{ij}x_j + B_jx^j$ is inhomogeneous.

The purpose of this section is to show that the result of \ref{sec:wick_ab} can be used inductively to calculate a Wick integral of a nondegenerate inhomogeneous quadratic form $Q(x)$ over a polyhedron $P$ (possibly unbounded) in $\RR^n$ and to get a flavor for the form of the answer.  We shall work with a particular situation that arises in Section \ref{sec:count_upp_half}.  Furthermore, we shall assume that $P$ is determined by inequalities of the form
\begin{align}
  0 \leq u_i + \sum_jR_{ij}x_j
\end{align}
for $i = 1, \dots, n + 1$.

Because $A$ is a symmetric matrix, it can be diagonalized by an orthogonal transformation.  This will produce a linear combination of integrals of the form
\begin{align}
  \int_Py_1^{k_1} \dots y_n^{k_n} e^{-\alpha_1 y^2_1/2 + \beta_1 y_1}\dots e^{-\alpha_n y^2_r/2 + \beta_n y_r}\, dy,
\end{align}
where $P$ is some new polyhedron.  Decomposing $P$ into subsets $P'$ that can be expressed in the form
\begin{align}
  a_1(y_2, \dots, y_n) &\leq y_1 \leq b_1(y_2, \dots, y_n)\\
  \vdots\nonumber\\
  a_{n - 1}(y_n) &\leq y_{n - 1} \leq b_{n - 1}(y_n)\nonumber\\  
  a_n &\leq y_n \leq b_n \nonumber
\end{align}
where $a_j(y_{j + 1}, \dots, y_n)$ and $b_j(y_{j + 1}, \dots, y_n)$ are linear in $y_{j + 1}, \dots, y_n$ with the possibility that $a_j = -\infty$ or $b_j = \infty$, and $a_n$ and $b_n$ are constants, possibly equal to $-\infty$ or $\infty$.

Decompose the integral over $P'$ as 
\begin{align}
  \int_{P_2'}\int_{a_1}^{b_1} y_1^{k_1} \dots y_n^{k_n} e^{-\alpha_1 y^2_1/2}\dots e^{-\alpha_n y^2_n/2} e^{\beta_1 y_1}\dots e^{\beta_n y_n}\, dy_1 \dots dy_n,
\end{align}

We apply Wick's theorem in one variable to $y_1$ to get a linear combination of elements of the form
\begin{align}
    \int_{P_2'} y_2^{k_2} \dots y_n^{k_n} e^{-\alpha_2 y_2^2/2} \dots e^{-\alpha_n y^2_n/2} e^{\beta_2 y_2}\dots e^{\beta_n y_n}J_{l, \alpha_1, \beta_1}(a_1, b_1)\, dy_2\dots dy_n 
\end{align}
for $l < k_1$ with an element of the form
\begin{align}
    \int_{P_2'} y_2^{k_2} \dots y_n^{k_n} e^{-\alpha_2 y_2^2/2} \dots e^{-\alpha_n y^2_n/2} e^{\beta_2 y_2}\dots e^{\beta_n y_n}I_{0, \alpha_1, \beta_1}(a_1, b_1)\, dy_2\dots dy_n 
\end{align}

By substituting the definition of $J_{l, \alpha_1, \beta_1}$, terms of the first form are equal to
\begin{align}
  &\int_{P_2'} y_2^{k_2} \dots y_n^{k_n} e^{-\alpha_2 y_2^2/2} \dots e^{-\alpha_n y^2_n/2} e^{\beta_2 y_2}\dots e^{\beta_n y_n} b_1^le^{-\alpha_1 b_1^2 + \beta_1 b_1} \, dy_2\dots dy_n \\
  -&\int_{P_2'} y_2^{k_2} \dots y_n^{k_n} e^{-\alpha_2 y_2^2/2} \dots e^{-\alpha_n y^2_n/2} e^{\beta_2 y_2}\dots e^{\beta_n y_n}a_1^le^{-\alpha_1 a_1^2 + \beta_1 a_1} \, dy_2\dots dy_n 
\end{align}

Let us focus our attention on the terms involving $b_1= d_2y_2 + \dots + d_ny_n$.  The analysis for the terms involving $a_1$ is similar.  Note that
\begin{align}
  &\alpha_2y_2^2 + \dots + \alpha_ny_n^2 + \alpha_1(d_2y_2 + \dots + d_ny_n)^2
\end{align}
has rank $n - 1$.  To prove this, let $d$ be the thought of as a column vector with entries $d_2, \dots, d_r$.  Let $c = \sqrt{\alpha_1}d$ and let $A = \diag(\alpha_2, \dots, \alpha_r)$.  Then
\begin{align}
  \det(A + \alpha_1 dd^t) &= \det(A)\det(I + A^{-1}cc^t) \\
  &= \det(A)(1 + c^tA^{-1}c) = \det(A)(1 + |\sqrt{A^{-1}}c|^2) > 0
\end{align}
which implies that the quadratic form in the variables $y_2, \dots, y_r$ is nondegenerate.

Regarding the term involving $I_{0, \alpha_1, \beta_1}(a_1, b_1)$, there does not seem to be any further simplification, unless $a_1 = -\infty$ and $b_1 = \infty$, so, in general, we leave the answer in this form.

\section{Heat Kernel Counter Terms}
\label{sec:heat_kern_count}

\subsection{A Motivating Example}\label{sec:some-examples}
Due to the complexity of the renormalization procedure for a general Feynman graph, it is helpful to begin with a particular Feynman graph in a particular theory.  We will work in the $\phi^4$ theory, i.e.\ the scalar field theory theory with classical interaction
\begin{align}
  I(\phi) = \frac{1}{4!}\int_M \phi^4.
\end{align}

Nothing needs to be done to renormalize a Feynman weight at the $0$-loop level, since the limit
\begin{align}
  \lim_{\epsilon \to 0^+} w_\gamma(P_\epsilon^L, I)
\end{align}
already exists for any tree $\gamma$.

To illustrate what happens at the higher loop level, we will work with the $1$-loop graph $\gamma$
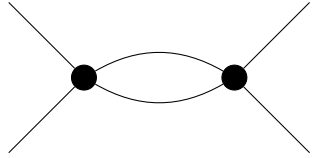
\begin{figure}[H]
  \centering
  \begin{tikzpicture}
    \node[style={circle,draw, fill=black}] (A) at (-1, 0) {};
    \node[style={circle,draw, fill=black}] (B) at (1, 0) {};
    \node (L1) at (-2, -1){};
    \node (L2) at (-2, 1){};
    \node (L3) at (2, -1){};
    \node (L4) at (2, 1){};
    
    \path (A) edge[bend right] node[auto,below] {} (B);
    \path (A) edge[bend left] node[auto,above] {} (B);
    \draw (A) --  (L1);
    \draw (A) --  (L2);
    \draw (B) --  (L3);
    \draw (B) --  (L4);

  \end{tikzpicture}.
\end{figure}
\noindent The Feynman weight $w_\gamma(P_\epsilon^L, I)$ is computed by labelling the vertices by the interaction $I$, the edges by the propagator $P_\epsilon^L$, and the tails by the input field $\phi$ 
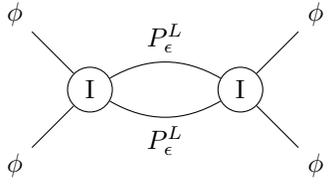
\begin{figure}[H]
  \centering
  \begin{tikzpicture}
    \node[style={circle,draw}] (A) at (-1, 0) {I};
    \node[style={circle,draw}] (B) at (1, 0) {I};
    \node (L1) at (-2, -1){$\phi$};
    \node (L2) at (-2, 1){$\phi$};
    \node (L3) at (2, -1){$\phi$};
    \node (L4) at (2, 1){$\phi$};
    
    \path (A) edge[bend right] node[auto,below] {$P_\epsilon^L$} (B);
    \path (A) edge[bend left] node[auto,above] {$P_\epsilon^L$} (B);
    \draw (A) --  (L1);
    \draw (A) --  (L2);
    \draw (B) --  (L3);
    \draw (B) --  (L4);

  \end{tikzpicture}
\end{figure}
\noindent and then contracting.  As defined in (\ref{eq:req_prop}), the regularized propagator $P_\epsilon^L$ is the integral of the heat kernel $K_t$ over the time interval $[\epsilon, L]$.  The Feynman weight of $\gamma$ is therefore given by
\begin{align}\label{eq_wgamma}
  w_\gamma(P_\epsilon^L, I)[\phi] = \int_{[\epsilon, L]^2} f_{\gamma, I}(t_1, t_2)[\phi]\, dt_1dt_2
\end{align}
where
\begin{align}
  f_{\gamma, I}(t_1, t_2)[\phi] = \int_{M^2} K_{t_1}(x_1, x_2) K_{t_2}(x_1, x_2) \phi(x_1)^2 \phi(x_2)^2 .
\end{align}
We make the definition $\Phi(x_1, x_2) = \phi(x_1)^2\phi(x_2)^2$ to avoid unecessary detail in subsequent equations.

In the case $M = \RR^n$, the heat kernel is given by
\begin{align}
  K_t(x_1, x_2) = (4\pi t)^{-n/2} e^{-|x_1 - x_2|^2/ 4t}.
\end{align}
Making the substitution for $K_t$ and the change of variables
\begin{align}
  w &= x_1 + x_2\\
  y &= x_1 - x_2
\end{align}
we have
\begin{align}
  f_{\gamma, I}(t_1, t_2)[\phi] = C(t_1t_2)^{-n/2} \int_{(\RR^n)^2} e^{-|y|^2(1/4t_1 + 1/4t_2)} \Phi(w, y).
\end{align}
for some constant $C$.  Let $\Phi^N(w, y)$ be the Taylor polynomial of degree $N$ of $\Phi(w, \cdot)$ and define
\begin{align}
  f^N_\gamma(t_1, t_2)[\phi] = C(t_1t_2)^{-n/2} \int_{(\RR^n)^2} e^{-|y|^2(1/4t_1 + 1/4t_2)}\Phi^N(w, y)
\end{align}
Costello's idea is, roughly, to define
\begin{align}
  w_\gamma^N(P_\epsilon^L, I) = \int_{[\epsilon, L]^2} f^N_{\gamma, I}(t_1, t_2)[\phi]\, dt_1dt_2.
\end{align}
We would then hope that by making $N$ sufficiently large, we can sufficiently control the error $|f_{\gamma, I}(t_1, t_2)[\phi] - f_{\gamma, I}^N(t_1, t_2)[\phi]|$
to force the limit $\lim_{\epsilon \to 0^+} [w_\gamma(P_\epsilon^L, I) - w_\gamma^N(P_\epsilon^L, I)]$ to exist.  

This is the idea in spirit, but there are additional subtleties needed to ensure we can sufficiently bound the error.  In fact, we will need to break up $[\epsilon, L]^2$ into regions and have different approximations and error bounds on each region.  Firstly, we consider regions for each ordering of the set of edges of $\gamma$.  For the Feynman graph we are working with, we consider $t_1 \leq t_2$ and $t_2 \leq t_1$.  This particular graph is symmetric with respect to interchange of the edges, so, we can in fact assume without loss of generality that $t_1 \leq t_2$.  With this assumption in place, we further subdivide the region.  Choose a positive constant $R > 2$ (the reason for this lower bound on $R$ will become clear below) and consider the two regions $t_2^R \leq t_1$ and $t_1 \leq t_2^R$.

If $t_2^R \leq t_1$, we have $(t_1t_2)^{-n/2} \leq t_2^{-(R + 1)n/2}$ and therefore we have the error bound
\begin{align}
  |f_{\gamma, I}(t_1, t_2)[\phi] - f_{\gamma, I}^N(t_1, t_2)[\phi]| &\lesssim t_2^{-(R + 1)n/2}\int_{\RR^n} e^{-|y|^2 / 2t_2}|y|^{N + 1}\\
  &\lesssim t_2^{-(R + 1)n/2}t_2^{(N + 1)/2 + n/2} = t_2^{(N + 1)/2 - Rn/2} 
\end{align}
and then by making $N$ large enough we can ensure that $(N + 1)/2 -  Rn/2 \geq 0$.  Let $N_1$ be such a sufficiently large integer such that for $N = N_1$, the error is bounded by some constant times a non-negative power of $t_2$.

Before considering the region $t_1 \leq t_2^R$, we elucidate the structure of $f^N_{\gamma, I}(t_1, t_2)[\phi]$.  We write $\Phi^N(w, y) = \sum_{|K| \leq N} c_K(w)y^K$ so that 
\begin{align}
  f^N_{\gamma, I}(t_1, t_2)[\phi] = C(t_1t_2)^{-n/2} \sum_{|K| \leq N}   \left(\int_{\RR^n} e^{-|y|^2(1/4t_1 + 1/4t_2)}y^K \right) \int_{\RR^n} c_K.
\end{align}
Applying Theorem \ref{sec:gau_sev}, Wick's theorem on $\RR^n$, to the integral over $y$, we have
\begin{align}
  f^N_{\gamma, I}(t_1, t_2)[\phi] = \sum_{|K| \leq N} F_K(t_1, t_2) \int_{\RR^n}c_K.
\end{align}
Note that for each $K$, $F_K(t_1, t_2)$ is a sum of square roots of rational functions of $t_1$ and $t_2$ and $\int_{\RR^n}c_K$ is a local functional of $\phi$.

If $t_1 \leq t_2^R$, we do something slightly different.  We begin by choosing the subgraph $\gamma'$ of $\gamma$ corresponding to the edge labelled by $t_1$
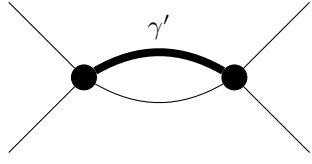
\begin{figure}[H]
  \centering
  \begin{tikzpicture}
    \node[style=circle,draw,fill=black] (A) at (-1, 0) {};
    \node[style=circle,draw,fill=black] (B) at (1, 0) {};

    \path[thin] (A) edge[bend right] (B);
    \path[line width=3pt] (A) edge[bend left] node[fill=none, above]{$\gamma'$} (B);
    \draw (A) -- (-2, -1);
    \draw (A) -- (-2, 1);
    \draw (B) -- (2, -1);
    \draw (B) -- (2, 1);
  \end{tikzpicture}.
\end{figure}
\noindent and treating the edges outside $\gamma'$ as pairs of input tails.  That is, 
\begin{align}
  f_{\gamma', \gamma, I}(t_1, t_2)[\phi] &:= f_{\gamma, I}(t_1, t_2)[\phi]
\end{align}
but now group the terms differently so that
\begin{align}
  f_{\gamma', \gamma, I}(t_1, t_2)[\phi] &=  C(t_1t_2)^{-n/2}\int_{(\RR^n)^2} e^{-|y|^2/4t_1}\Psi(w, y, t_2)
\end{align}
where $\Psi(w, y, t_2) = e^{-|y|^2/4t_2}\Phi(w, y)$ and $C$ is some constant.  Define
\begin{align}
  f^N_{\gamma', \gamma}(t_1, t_2)[\phi] =  C(t_1t_2)^{-n/2}\int_{(\RR^n)^2} e^{-|y|^2/4t_1}\Psi^N(w, y, t_2)
\end{align}
where $\Psi^N(w, y, t_2)$ is the order $N$ Taylor polynomial of $\Psi(w, \cdot, t_2)$.  Then because the $(N + 1)$-st derivative in the $y$ variable of $\Psi(w, y, t_2)$ is bounded by a constant times $t_2^{-(N + 1)}$, we have the bound
\begin{align}
  |f_{\gamma', \gamma, I}(t_1, t_2)[\phi] - f^N_{\gamma', \gamma, I}(t_1, t_2)[\phi]| &\lesssim (t_1t_2)^{-n/2}t_2^{-(N + 1)}\int_{\RR^n} e^{-|y|^2 / t_1}|y|^{N + 1}\\
                                                                                      &\lesssim (t_1t_2)^{-n/2}t_2^{-(N + 1)} t_1^{(N + 1)/2 + n/2} \\
                                                                                      &\lesssim  t_2^{-n/2 - (N + 1)}t_2^{(R/2)(N + 1)}\\
  &\lesssim t_2^{\left(R/2 - 1\right)(N + 1) - n/2}.
\end{align}
This is why we require that $R > 2$, so that for $N$ sufficiently large $\left(R/2 - 1\right)(N + 1) - n/2 \geq 0$.  Let $N_2$ be such an $N$.

The counterterm is given by 
\begin{align}
  w^{\text{ct}}_\gamma(P_\epsilon^L, I) = \int \limits_{\substack{\epsilon \leq t_1 \leq t_2 \leq L \\ t_2^R \leq t_1}} f_{\gamma, I}^{N_1}(t_1, t_2)[\phi] + \int\limits_{\substack{\epsilon \leq t_1 \leq t_2 \leq L \\ t_1 \leq t_2^R }} f_{\gamma', \gamma, I}^{N_2}(t_1, t_2)[\phi] 
\end{align}
plus the same two terms with $t_1$ and $t_2$ permuted.
By construction, the limit
\begin{align}
  \lim_{\epsilon \to 0^+} [w_\gamma(P_\epsilon^L, I) - w^{\text{ct}}_\gamma(P_\epsilon^L, I)]
\end{align}
exists, as desired.

In the case of the Euclidean half space $\HH^n$, the Dirichlet heat kernel is given by
\begin{align}
  K_t(x_1, x_2) = (4\pi t)^{-n/2}[e^{-|x_1 - x_2|^2/4t} - e^{-|x_1 - x_2^*|^2/4t}]
\end{align}
and the Neumann heat kernel is given by
\begin{align}
  K_t(x_1, x_2) = (4\pi t)^{-n/2}[e^{-|x_1 - x_2|^2/4t} + e^{-|x_1 - x_2^*|^2/4t}]
\end{align}
where $x_2^{*}$ is the reflection through the boundary.  We give a combined analysis of the procedure for both the Dirichlet and Neumann heat kernel by absorbing any signs into constants in the formulas.  We write
\begin{align}
  w_\gamma(P_\epsilon^L, I) &=  \int_{[\epsilon, L]^2} \int_{(\HH^4)^2} K_{t_1}(x_1, x_2) K_{t_2}(x_1, x_2) \Phi(x_1, x_2)
\end{align}
as
\begin{align}
  \int_{[\epsilon, L]^2} \left[ f_{\gamma, I; 0, 0}(t_1, t_2)[\phi] + f_{\gamma, I; 1, 0}(t_1, t_2)[\phi] + f_{\gamma, I; 0, 1}(t_1, t_2)[\phi] + f_{\gamma, I; 1, 1}(t_1, t_2)[\phi]\right]
\end{align}
where
\begin{align}
  f_{\gamma, I; \beta_1, \beta_2}(t_1, t_2)[\phi] &= C(t_1t_2)^{-n/2}\int_{(\HH^4)^2} e^{-d_{\beta_1}(x_1, x_2)^2/4t_1 - d_{\beta_2}(x_1, x_2)^2/4t_2}\Phi(x_1, x_2)
\end{align}
where $C$ is some constant and
\begin{align}
  d_\beta(x_1, x_2)^2 =
  \begin{cases}
    |x_1 - x_2|^2 & \text{if $\beta = 0$}\\
    |x_1 - x^*_2|^2 &\text{if $\beta = 1$}\\
  \end{cases}
\end{align}

First consider the case $(\beta_1, \beta_2) = (0, 0)$.  As before, introduce the coordinates, $w = x_1 + x_2$ and $y = x_1 - x_2$.  Let $w = (\overline{w}, w_n)$, $y = (\overline{y}, y_n)$, $x_1 = (\overline{x}_1, x_{1, n})$, and $x_2 = (\overline{x}_2, x_{2, n})$. 
Note that because $x_{1, n} \geq 0$ and $x_{2, n} \geq 0$, we have $w_n \geq 0$ and $-w_n \leq y_n \leq w_n$. 
We have
\begin{align}
  -d^2_{\beta_1}(w, y)/4t_1 - d^2_{\beta_2}(w, y)/4t_2 = -|y|^2(1/4t_1 + 1/4t_2).
\end{align}
Let $\Phi^N$ be the Taylor expansion to order $N$ at 0 of $\Phi$ in $y$.
Then by definition
\begin{align}
  f_{\gamma, I; 0, 0}^N(t_1, t_2)[\phi] = C(t_1t_2)^{-n/2}\int_{\HH^n_w} \int_{[-w_n, w_n]_{y_n}}\int_{\RR^3_{\overline{y}}} e^{-|y|^2(1/4t_1 + 1/4t_2)} \Phi^N(w, y)
\end{align}
For $t_2^R \leq t_1$, we therefore we get the same bound
\begin{align}
  |f_{\gamma, I; 0, 0}(t_1, t_2)[\phi] - f_{\gamma, I; 0, 0}^N(t_1, t_2)[\phi]| \lesssim t_2^{(N + 1)/2 - Rn/2} 
\end{align}
as in the case of $\RR^n$.

However, there is still more to do.  Upon examining the structure of $f^N_{\gamma, I, 0, 0}(t_1, t_2)[\phi]$, we find that we no longer have a summation of local integrals on $\HH^n$, each weighted by the square root of some rational function in $t_1$ and $t_2$, as in the case of $\RR^n$.  This is because we are only integrating $y_n$ on the interval $[-w_n, w_n]$ rather than $(-\infty, \infty)$.  To remedy the situation, the idea is now that, where necessary, we can also take the Taylor expansion of the fields in $w_n$ and then calculate the integral over $w_n$ by applying Wick's theorem on the half line.  This will then produce the hoped for local integrals on the boundary $\RR^{n - 1}$.

More explicitly, suppose that $\Phi^N(w, y) = \sum_{|\overline{K}| + k \leq N} c_{\overline{K}, k}(w) \overline{y}^Ky_n^k$, and therefore that $f_{\gamma, I; 0, 0}^N(t_1, t_2)[\phi]$ is equal to 
\begin{align}\label{eq:bdrywght}
  \sum_{k = 0}^N\sum_{|\overline{K}| \leq N - k} F_{\overline{K}}(t_1, t_2) \int_{\HH^n_w} c_{\overline{K}, k} \int_{[-w_n, w_n]_{y_n}} e^{-y_n^2(1/4t_1 + 1/4t_2)} y_n^k.
\end{align}
Using that
\begin{align}
  \int_{[-w_n, w_n]_{y_n}} = \int_{(-\infty, \infty)_{y_n}} - \int_{[w_n, \infty)_{y_n}} - \int_{(-\infty, -w_n]_{y_n}}
\end{align}
we have
\begin{align}
  \int_{[-w_n, w_n]_{y_n}}e^{-y_n^2(1/4t_1 + 1/4t_2)}y_n^k &= \frac{\sqrt{2\pi}C_k}{(1/4t_1 + 1/4t_2)^{(k + 1)/2}}\\
  &- 2\int_{[w_n, \infty)_{y_n}}e^{-y_n^2(1/4t_1 + 1/4t_2)}y_n^k
\end{align}
where recall that $C_k = (k - 1)!!$ for $k$ even and $C_k = 0$ for $k$ odd.

Thus for $f_{\gamma, I; 0, 0}^N(t_1, t_2)[\phi]$ we have two types of terms.  For the first type of term, we are finished, having already produced a sum of local functionals on $\HH^n$.  For the second type of term we define $c_K^m = \sum_{i = 0}^m d_{K, i}(\overline{w}) w_n^i$, the Taylor expansion in $w_n$ to order $m$ of $c_K := c_{\overline{K}, k}$.  Then we define $f^{N, m}_{\gamma, I, 0, 0}(t_1, t_2)[\phi]$ to be the result of substituting $c_K^m$ for $c_K$ in $f^N_{\gamma, I, 0, 0}(t_1, t_2)[\phi]$.
Note that
\begin{align}
  w_n^i\int_{[w_n, \infty)_{y_n}}e^{-y_n^2(1/4t_1 + 1/4t_2)}y_n^k
\end{align}
is integrable over $w_n$ over the half line $\RR_{\geq 0}$.
Therefore we have
\begin{align}
  f^{N, m}_{\gamma, I, 0, 0}(t_1, t_2)[\phi] &= \sum_{|K| \leq N} \tilde{F}_K(t_1, t_2)\int_{\HH^n_w} c_K + \sum_{|K| \leq N} \sum_{i = 0}^m G_{K, i}(t_1, t_2) \int_{\RR_{\overline{w}}^{n - 1}} d_{K, i},
\end{align}
with the $\tilde{F}_K(t_1, t_2)$ being sums of square roots of rational functions and the $G_{K, i}(t_1, t_2)$ not seeming to have a closed form expression.  We claim that by making $m$ sufficiently large we can ensure that $|f^N_{\gamma, I, 0, 0}(t_1, t_2)[\phi] - f^{N, m}_{\gamma, I, 0, 0}(t_1, t_2)[\phi]|$ is bounded by a constant times a sufficiently large power of $t_2$.  More details about this will be given in Section \ref{sec:count_upp_half}.

The cases $(\beta_1, \beta_2) = (1, 0)$ and $(\beta_1, \beta_2) = (0, 1)$ are similar so we just consider $(\beta_1, \beta_2) = (1, 0)$.  We find that $f^N_{\gamma, I, 1, 0}(t_1, t_2)[\phi]$ is equal to
\begin{align}
  \sum_{k = 0}^N\sum_{|\overline{K}| \leq N - k}F_{\overline{K}}(t_1, t_2) \int_{\HH^n_w} c_K e^{-w_n^2/4t_1}\int_{[-w_n, w_n]_{y_n}}  y_n^k e^{-y_n^2/4t_2}.
\end{align}
Wick's theorem on $[-w_n, w_n]$ implies that $f_{\gamma, I; 1, 0}^N(t_1, t_2)[\phi]$ is equal to 
\begin{align}
  &\sum_{|K| \leq N} F_K(t_1, t_2)\int_{\HH^n} c_K e^{-w_n^2/4t_1} \int_{[-w_n, w_n]_{y_n}}e^{-y_n^2/4t_2}\\
  + &\sum_{|K| \leq N}\sum_{l = 0}^kF_{l, K}(t_1, t_2) \int_{\HH^n_w}c_K w_n^le^{-w_n^2(1/4t_1 + 1/4t_2)}.
\end{align}
For both terms, we replace $c_K$ with its Taylor expansion in $w_n$ to order $m$ and then carry out the integral over $w_n$ over $\RR_{\geq 0}$.

Lastly, in the case $(\beta_1, \beta_2) = (1, 1)$, we can introduce the coordinates, $\overline{w} = \overline{x}_1 + \overline{x}_2$, $\overline{y} = \overline{x}_1 - \overline{x}_2$, $w_n = x_{1, n} - x_{2, n}$, and $y_n = x_{1, n} + x_{2, n}$.  Let $w = (\overline{w}, w_n)$, $y = (\overline{y}, y_n)$, $x_1 = (\overline{x}_1, x_{1, n})$, and $x_2 = (\overline{x}_2, x_{2, n})$.  We then find we are in the same situation as when $(\beta_1, \beta_2) = (0, 0)$

On the set where $t_1 \leq t_2^R$, we introduce the notation $f_{\gamma', \gamma, I, \beta_1, \beta_2}(t_1, t_2)[\phi]$ for $f_{\gamma, I,  \beta_1, \beta_2}(t_1, t_2)[\phi]$ which we write as
\begin{align}
  C(t_1t_2)^{-n/2}\int_{(\HH^n)^2} e^{-d_{\beta_1}(x_1, x_2)^2/4t_1} \Psi_{\beta_2}(x_1, x_2, t_2)
\end{align}
where
\begin{align}
  \Psi_{\beta_2}(x_1, x_2, t_2) = e^{-d_{\beta_2}(x_1, x_2)^2/4t_2}\Phi(x_1, x_2)
\end{align}
and construct $f^N_{\gamma', \gamma, I, \beta_1, \beta_2}(t_1, t_2)[\phi]$ similarly to the way we did on $\RR^n$ by Taylor expanding $\Psi_{\beta_2}$ in $\overline{y}$ and $y_n$.  We can bound the error by a similar procedure, but as in the case of $t_2^R \leq t_1$, $f^N_{\gamma', \gamma, I; \beta_1, \beta_2}(t_1, t_2)[\phi]$ will not be a sum of local integrals of $\phi$ each multiplied by the square root of a rational function in $t_1$ and $t_2$.  So, we will need to remedy this when necessary with an additional Taylor expansion in $w_n$ and application of Wick's theorem.
More details about this will be given in Section \ref{sec:count_upp_half}.

\subsection{Covering \texorpdfstring{$(0, \infty)^{|E(\gamma)|}$}{(0, infty) prod |E(gamma)|}}
\label{sec:covering_egamma}
When working with a more general Feynman graph, it becomes necessary to cover $(0, \infty)^{|E(\gamma)|}$ and in particular $(\epsilon, L)^{|E(\gamma)|}$ by sets generalizing those we used in the previous section when we had $|E(\gamma)| = 2$.

We shall assume that $|E(\gamma)| \geq 1$.  We denote $\mathbf{t} = (t_1, \dots, t_k)$.  For each permutation $\sigma$ of the set of $|E(\gamma)|$ elements, there is a subset of $(0, \infty)^{|E(\gamma)|}$
\begin{align}
  S_\sigma = \{\mathbf{t} \in (0, \infty)^{|E(\gamma)|} : t_{\sigma(1)} < \dots < t_{\sigma(|E(\gamma)|)}\}.
\end{align}
and it is clear that
\begin{align}
  \cup_{\sigma \in S_k} \overline{S_\sigma} = (0, \infty)^{|E(\gamma)|}.
\end{align}

The procedure we are about to describe should be applied separately within each of the $S_\sigma$, but we fix
\begin{align}
  S_{\text{id}} = \{\mathbf{t} \in (0, \infty)^{|E(\gamma)|} : t_1 < \dots < t_{|E(\gamma)|}\}.
\end{align}
for notational clarity.  Replace $t_i$ with $t_{\sigma(i)}$ to recover the case of a general permutation $\sigma$.

We assume that $R > 1$. 
\begin{defin}
  For $i, j \in \{1, \dots, |E(\gamma)|\}$ with $i < j$, define
  \begin{align}
    B_R^{i, j} &= \{ \mathbf{t} \in S_{\text{id}} :  t_i < t_j^R \}.
  \end{align}
  and define
  \begin{align}
    C_R^{i, j} = \{ \mathbf{t} \in S_{\text{id}} :  t_j^R < t_i \}.
  \end{align}

And lastly for $j \in \{2, \dots, |E(\gamma)| - 1\}$, define
\begin{align}
  A_R^j &= B_R^{j, j + 1} \cap C_R^{1, j}\\
  &= \{\mathbf{t} \in S_{\text{id}} : t_j < t_{j + 1}^R \text{ and } t_j^R < t_1\}
\end{align}
and let $A_R^{1} = B_R^{1, 2}$ and $A^{|E(\gamma)|} = C_R^{1, |E(\gamma)|}$.
\end{defin}
Note that $B_R^{i, j} = S_{\text{id}} \setminus \overline{C_R^{i, j}}$.  A couple of facts about the sets $B_R^{i, j}$ and $C_R^{i, j}$ are collected in the following proposition:
\begin{prop}
  For $i_1 < i_2 < i_3$.
  \begin{align}
    B_R^{i_1, i_2} \cap  B_S^{i_2, i_3} \subset B_{RS}^{i_1, i_3}  
  \end{align}
  and similarly
  \begin{align}
    C_R^{i_1, i_2} \cap  C_S^{i_2, i_3} \subset C_{RS}^{i_1, i_3}
  \end{align}
\end{prop}
\begin{proof}
  If $\mathbf{t} \in C_R^{i_1, i_2} \cap C_S^{i_2, i_3}$, then $t_{i_2}^R < t_{i_1}$ and $t_{i_3}^S < t_{i_2}$.  This implies that
  \begin{align}
    t_{i_3}^{RS} < t_{i_1}.
  \end{align}
  i.e.\ $\mathbf{t} \in C_{RS}^{i_1, i_3}$.
  The proof of the first inclusion is similar.
\end{proof}

The following two statements are trivially true:
\begin{prop}
  For $i, j \in \{1, \dots, |E(\gamma)|\}$ with $i < j$, define
  \begin{align}
    \tilde{B}_R^{i, j} = \{ \mathbf{t} \in S_{\text{id}} : t_\alpha < t^R_\beta\text{, for $\alpha \leq i$ and $j \leq \beta$} \}.
  \end{align}
  and
  \begin{align}
    \tilde{C}_R^{i, j} = \{ \mathbf{t} \in S_{\text{id}} :  t_\alpha^R < t_\beta \text{, for $\alpha \leq i$ and $j \leq \beta$} \}.\}.
  \end{align}
  Then $\tilde{B}_R^{i, j} = B_R^{i, j}$ and $\tilde{C}_R^{i, j} = C_R^{i, j}$
\end{prop}

\begin{prop}
  For $j_1 \leq j_2$, $C_R^{i, j_2} \subseteq C_R^{i, j_1}$.
\end{prop}

The next two proposions are needed to prove Theorem \ref{thm:EIdisjoint}.
\begin{prop}\label{prop:c_r_d_r}
  $C_R^{i, j} \cap B_R^{k, l} = \emptyset$ for $i \leq k$ and $l \leq j$.
\end{prop}
\begin{proof}
  If $t_j^R < t_i$ and $t_k < t_l^R$, then
  \begin{align}
    t_i \leq t_k < t_l^R \leq t_j^R < t_i,
  \end{align}
  a contradiction.
\end{proof}

\begin{defin}\label{def:E_RI}
  Fix some sequence of positive integers $s_0, \dots, s_m$.  We consider sequences $I$ of the form $1 = i_0 < i_1 < \dots < i_m \leq E(\gamma)$, where $m \leq E(\gamma) - 1$.  For any such sequence $I$, we define the sets $E^I_R = \cap_{j = 0}^m E^I_{R, j}$, where $E^I_{R, j}$ is defined such that for $j < m$
  \begin{align}
    E^I_{R, j} =
    \begin{cases}
      S_{id} & \text{if $j = 0$} \\
      C^{i_{j - 1}, i_{j}}_{R^{s_j}} \cap B^{i_{j - 1}, i_j + 1}_{R^{s_j}} & \text{otherwise}
    \end{cases}
  \end{align}
  and for $j = m$
  \begin{align}
    E^I_{R, m} = 
    \begin{cases}
      B_{R^{s_1}}^{1, 2} &\text{if $m = 0$}\\
      C^{i_{m - 1}, i_m}_{R^{s_m}} \cap B^{i_{m - 1}, i_m + 1}_{R^{s_m}} \cap B_{R^{s_{m + 1}}}^{i_m, i_m + 1} & \text{if $i_m \neq |E(\gamma)|$ and $m > 0$}\\
      C^{i_{m - 1}, i_m}_{R^{s_m}} & \text{if $i_m = |E(\gamma)|$}.
    \end{cases}
  \end{align}

\end{defin}
\begin{thm}\label{thm:closures_cover}
  The closures $\overline{E}_R^I$ form a cover of $(0, \infty)^{|E(\gamma)|}$. 
\end{thm}
\begin{proof}
  If $t_1 \leq t_2^{R^{s_1}}$, then $\mathbf{t} \in \overline{B}^{1, 2}_R$.  Thus letting $m = 0$, we have $\mathbf{t} \in \overline{E}_R^I$.

  Otherwise, let $i_1$ be the largest integer such that $t^{R^{s_1}}_{i_1} \leq t_1 = t_{i_0}$.  Then $\mathbf{t} \in \overline{C}^{i_0, i_1}_{R^{s_1}}$.  If $i_1 = |E(\gamma)|$, letting $m = 1$, we have $\mathbf{t} \in \overline{E}_R^I$.  If $i_1 < |E(\gamma)|$, then $\mathbf{t} \in B^{i_0, i_1 + 1}_{R^{s_1}}$.  If $t_{i_1} \leq t^R_{i_1 + 1}$ then $\mathbf{t} \in \overline{B}^{i_1, i_1 + 1}_{R^{s_2}}$ and letting $m = 1$, we have $\mathbf{t} \in \overline{E}_R^I$.

  Otherwise, let $i_2$ be the largest integer such that $t^{R^{s_2}}_{i_2} \leq t_{i_1}$.  Then $i_2 \in \overline{C}^{i_1, i_2}_{R^{s_2}}$.  If $i_2 = |E(\gamma)|$, letting $m = 2$, we have $\mathbf{t} \in \overline{E}_R^I$.  If $i_2 < |E(\gamma)|$, then $\mathbf{t} \in B^{i_1, i_2 + 1}_{R^{s_2}}$.  If $t_{i_2} \leq t^{R^{s_3}}_{i_2 + 1}$ then $\mathbf{t} \in \overline{B}^{i_2, i_2 + 1}_{R^{s_3}}$ and letting $m = 2$, we have $\mathbf{t} \in \overline{E}_R^I$.
  
  And so on \dots
\end{proof}

\begin{thm}\label{thm:EIdisjoint}
  The sets $E^I_R$ are disjoint.
\end{thm}
\begin{proof}
  We prove this by induction.  Consider the distinct sequences $1 = i_0 < i_1 < \dots < i_m \leq E(\gamma)$ and $1 = j_0 < j_1 < \dots < j_n \leq E(\gamma)$, where without loss of generality we assume that $m \leq n$.  

Suppose that $i_l \neq j_l$, but $i_1 = j_1, \dots, i_{l - 1} = j_{l - 1}$.  Then
\begin{align}
  E^I_{R, l} \cap E^J_{R, l} &\subseteq C^{i_{l - 1}, i_{l}}_{R^{s_l}} \cap B^{i_{l - 1}, i_l + 1}_{R^{s_l}} \cap C^{i_{l - 1}, j_{l}}_{R^{s_l}} \cap B^{i_{l - 1}, j_l + 1}_{R^{s_l}} \\
  &= \emptyset. \nonumber
\end{align}
because $C_R^{i, j} \cap D_R^{i, k} = \emptyset$ for $k \leq j$ by Proposition \ref{prop:c_r_d_r}.

It is also possible that $i_1 = j_1, \dots, i_m = j_m$, but $m < n$.  Then
\begin{align}
  E^I_{R, m} \cap E^J_{R, m + 1} &\subseteq  B_{R^{s_{m + 1}}}^{i_m, i_m + 1} \cap C^{i_m, j_{m +1}}_{R^{s_{m + 1}}} \\
  &= \emptyset.\nonumber
\end{align}
again by Proposition \ref{prop:c_r_d_r}.
\end{proof}

Now specialize to a specific sequence $s_0 = 1$ and $s_i = 2^{i - 1}$ for $i > 0$.
\begin{thm}\label{thm:ERcontainedinAi}
  Consider the sequence $1 = i_0 < i_1 < \dots < i_m \leq |E(\gamma)|$.  Then
  \begin{align}
    E_R^I \subseteq A^{i_m}_{R^{2^{m}}}
  \end{align}
\end{thm}
\begin{proof}
  If $m = 0$, it is clear that $E_R^I \subseteq A^1_R = B^{1, 2}_R$.
  
  If $m > 0$, 
  \begin{align}
      E_R^I &\subseteq
            \begin{cases}
              C_R^{i_0, i_1} \cap C_R^{i_1, i_2} \cap C_{R^2}^{i_2, i_3} \dots \cap C_{R^{2^{m - 1}}}^{i_{m - 1}, i_m} \cap B^{i_m, i_m + 1}_{R^{2^m}}&\text{if $i_m < |E(\gamma)|$}\\
              C_R^{i_0, i_1} \cap C_R^{i_1, i_2} \cap C_{R^2}^{i_2, i_3} \dots \cap C_{R^{2^{m - 1}}}^{i_{m - 1}, i_m} &\text{if $i_m = |E(\gamma)|$}
            \end{cases}\\
    &\subseteq
      \begin{cases}
        C^{i_0, i_m}_{R^{2^m}}\cap B^{i_m, i_m + 1}_{R^{2^m}}&\text{if $i_m < |E(\gamma)|$}\\
        C^{i_0, i_m}_{R^{2^m}} &\text{if $i_m = |E(\gamma)|$}
      \end{cases}\\
    &=A^{i_m}_{R^{2^m}}
  \end{align}
\end{proof}

The construction of the counterterms in \ref{sec:induct-constr-count} will be based on a refinement of the covering $\{\overline{E}_R^I\}$.  For $l < |E(\gamma)|$, given a sequence $l = i_0 < \dots < i_m \leq |E(\gamma)|$, define analogously the sets $E_R^I$ forming a cover of the set
\begin{align}
  S_{\text{id}}^l = \{(t_l , \dots , t_{|E(\gamma)|}): t_l < \dots < t_{|E(\gamma)|}\}
\end{align}
For $l = 1$, we recover the sets $E_R^I$  which were defined earlier.  The following is a corollary of Theorem \ref{thm:closures_cover}:
\begin{cor}\label{cor:refined_cover}
  Consider the collection of sequences of the form
  \begin{align}
    1 &= i^{(1)}_0 < i^{(1)}_1 < \dots < i^{(1)}_{m^{(1)}}\\
    i^{(1)}_{m^{(1)}}  &= i^{(2)}_{0} < i^{(2)}_1 < \dots < i^{(2)}_{m^{(2)}} \nonumber \\
      &\dots \nonumber \\
     i^{(p - 1)}_{m^{(p - 1)}}  &= i^{(p)}_{0} < i^{(p)}_1 < \dots < i^{(p)}_{m^{(p)}} = |E(\gamma)|.\nonumber 
  \end{align}
  Then the sets
  \begin{align}
    \overline E^{I^{(1)}}_{R} \cap \overline E^{I^{(2)}}_{R}\dots \cap \overline E^{I^{(p)}}_{R}
  \end{align}
  form a cover of $S_{\text{id}}$.
\end{cor}

\subsection{Local Functionals and Feynman Weights}
\label{sec:local-funct-feynm}
\subsubsection{Differential Operators}
Let $M$ be a smooth manifold, let $E$ be a vector bundle and let $\underline{\RR}$ be the trivial line bundle.  Let $\mathcal{E} = \Gamma(E)$ and $C^\infty(M) = \Gamma(\underline{\RR})$.
  A differential operator $D:E \to \underline{\RR}$ is an $\RR$-linear map $\mathcal{E} \to C^\infty(M)$ which can be given locally as
\begin{align}
  s = \alpha^ie_i \mapsto a_j^I \frac{\partial \alpha^j}{\partial x^I} 
\end{align}
where $e_1, \dots, e_r$ is a local frame for $E$ on some sufficiently small coordinate neighborhood $U$, and $\alpha^1, \dots, \alpha^r$ and $a_i^I$ are functions on $U$.

Equivalently, a differential operator is a bundle map $\iota_D: J(E) \to \underline{\RR}$, where $J(E)$ is the jet bundle of $E$.  The differential operator $D$ is determined by $\iota_P$ by composing with the jet prolongation of $s$, $j(s): M \to J(E)$. That is, $D(s) = \iota_P \circ j(s)$.

\subsubsection{Local Functionals}

\begin{defin}
  A \emph{local functional} $\mathcal{O}^k_{\text{loc}}(\mathcal{E})$ of degree $k$ is a functional in $\mathcal{O}^k(\mathcal{E})$ of the form
  \begin{align}
    s \mapsto \sum_{i = 1}^m \int_M D_{i, 1}(s) \dots D_{i, k}(s)
  \end{align}
for some collection of differential operators $D_{i, j}: E \to \underline{\RR}$.
\end{defin}
Locally (in Einstein summation convention),
\begin{align}
  \alpha^i e_i \mapsto & \int_U a_{j_1, \dotsm j_k}^{I^1, \dots, I^k} \frac{\partial \alpha^{j_1}}{\partial x^{I^1}} \dots \frac{\partial \alpha^{j_k}}{\partial x^{I^k}}.
\end{align}
for a collection of functions $a_{j_1, \dots, j_k}^{I^1, \dots, I^k}$ on $U$.

\subsubsection{Evaluation of $w_\gamma(P,I)$}
We shall work in the scalar theory case for simplicity and assume that $E = \underline{\RR}$.

We would like to describe the form of $w_\gamma(P_\epsilon^L, I)$ where $I \in \mathcal{O}_{\text{loc}}(\mathcal{E})[[\hbar]]$ is a power series of local functionals and
\begin{align}
  P_\epsilon^L = \int_\epsilon^L K_t \, dt
\end{align}
where $K_t$ is the heat kernel of $M$.  
For each vertex $v \in V(\gamma)$, we associate the functional $I_{g(v), k(v)}$, where $k(v)$ is the valency of the vertex $v$.  Assume that within a given chart $U$, 
\begin{align}
  S^{k(v)}I_{g(v), k(v)}(\alpha_1, \dots, \alpha_{k(v)}) &= \int_U a^{I_v^1, \dots, I_v^{k(v)}} \frac{\partial \alpha_1}{\partial x^{I_v^1}} \dots \frac{\partial \alpha_{k(v)}}{\partial x^{I_v^{k(v)}}}.
\end{align}
where $I_v^1, \dots, I_v^{k(v)}$ are multi-indices with $|I_v^1| + \dots + |I_v^{k(v)}| \leq \ord I_{g(v), k(v)}$, where $\ord I_{g(v), k(v)}$ is the highest order of any multi-differential operator appearing in the functional $I_{g(v), k(v)}$.

We have implicitly chosen an ordering on the set of half edges incident on each vertex $v$ and an orientation on each edge.
Then $\gamma$ determines the maps
\begin{align}
  Q &: T(\gamma) \to  \bigsqcup_{v \in V(\gamma)} \{1, \dots, k(v)\} \\
  Q_1&: E(\gamma) \to \bigsqcup_{v \in V(\gamma)} \{1, \dots, k(v)\} \nonumber\\
  Q_2&: E(\gamma) \to \bigsqcup_{v \in V(\gamma)} \{1, \dots, k(v)\}\nonumber 
\end{align}
Also denote by $v_1(e)$ and $v_2(e)$ the first and second vertices of the edge $e$.  Similarly, let $v(h)$ denote the vertex of the tail $h$.  

With these data, we can give the expression
\begin{align}
  w_\gamma(P_\epsilon^L, I)[\alpha] = \int_{(\epsilon, L)^{|E(\gamma)|}} f_{\gamma, I}(\mathbf{t})[\alpha].
\end{align}
where for $M = \RR^n$, we have (in Einstein summation convention)
\begin{align}\label{eq:f_gamma_rn}
  f_{\gamma, I}(\mathbf{t})[\alpha] &= \int_{{\RR^n}^{|V(\gamma)|}} \prod_{v \in V(\gamma)} a^{I_v^1, \dots, I_v^{k(v)}}(x_v) \prod_{e \in E(\gamma)} \frac{\partial K_{t_e}(x_{v_1(e)}, x_{v_2(e)})}{\partial x^{I_{v_1(e)}^{Q_1(e)}}\partial x^{I_{v_2(e)}^{Q_2(e)}}} \prod_{h \in T(\gamma)} \frac{\partial \alpha(x_{v(h)})}{\partial x^{I_{v(h)}^{Q(h)}}}
\end{align}
If $M$ is a smooth manifold then choose a partition of unity subordinate to a cover of $M$.  Then $f_{\gamma, I}(\mathbf{t})[\alpha]$ is a sum of integrals of the form
\begin{align}\label{eq:f_gamma_cpct}
  \int_U \chi_U\prod_{v \in V(\gamma)} a^{I_v^1, \dots, I_v^{k(v)}}(x_v) \prod_{e \in E(\gamma)} \frac{\partial K_{t_{e}}(x_{v_1(e)}, x_{v_2(e)})}{\partial x^{I_{v_1(e)}^{Q_1(e)}}\partial x^{I_{v_2(e)}^{Q_2(e)}}} \prod_{h \in T(\gamma)} \frac{\partial \alpha(x_{v(h)})}{\partial x^{I_{v(h)}^{Q(h)}}}
\end{align}
where $\chi_U$ is the partition of unity function supported in the open set $U$ in the cover of $M^{n|V(\gamma)|}$ and $\alpha^i$ are the coordinates of $\alpha$ in $U$.  Note that in the local calculations that follow, we will absorb $\chi_U$ into $a^{I_v^1, \dots, I_v^{k(v)}}(x_v)$.

\subsection{Local Counterterms on a Flat Manifold: Preliminaries}
\label{sec:count_rn}
There are two cases.  In the first case, $M = \RR^n$, and we carry out the construction globally by requiring that the fields $\phi$ are Schwartz functions and employing Schwartz seminorms.  In the second case, we carry out the construction locally on $M$, a general flat manifold without boundary.  Essentially the only difference of the local story from the global $\RR^n$ story is that the integrand of the Feynman weight should be multiplied by a compactly supported function in $B^{|V(\gamma)|}$, where $B$ is a ball centered at the origin in $\RR^n$.  While globally the fields $\alpha$ will be assumed to be compactly supported smooth functions on $M$, locally they are smooth functions on $B$.  Going from the global expression for the Feynman weight to the local setup of this section will be carefully justified in Section \ref{sec:count-comp-manif}

On $\RR^n$, the heat kernel has the simple form
\begin{align}\label{eq:scalar_heat}
  K_t(x, y) = (4\pi t)^{-n/2}e^{-|x - y|^2/4t}.
\end{align}
Locally on a flat manifold without boundary $M$, we also make the assumption that locally $K_t(x, y)$ is of the above form.  The justification will once again be discussed in detail in Section \ref{sec:count-comp-manif}

\subsubsection{Derivatives of $K_t$}

\begin{prop}\label{thm:derivatives-k_t}
 For a multi-index $I = (i_1, \dots, i_n)$, $\frac{\partial K_t}{\partial x^I \partial x^{I'}}$ is equal to a polynomial in $x_1, \dots, x_n$, $y_1, \dots, y_n$ and $1/t$, call it $P_{I, I'}(x, y, 1/t)$, multiplied by $(4\pi t)^{-n/2}e^{-|x - y|^2/4t}$.  The degree in $1/t$ of this polynomial is $|I| + |I'|$.

\end{prop}
\subsubsection{Powers of $t$ in $w_\gamma(P_\epsilon^L, I)$}

Let $O(\gamma)$ be the sum of the orders of the local functionals $I_{g(v), k(v)}$ for all $v \in V(\gamma)$, where the order of a local functional is the order of the product of differential operators in the integrand.  Also, define $\Phi_J$ by setting $\sum_{-O(\gamma) \leq |J| \leq 0 } \mathbf{t}^J \Phi_J$ equal to 
\begin{align}
   \prod_{v \in V(\gamma)} a^{I_v^1, \dots, I_v^{k(v)}}(x_v) \prod_{e \in E(\gamma)} P_{{I_{v_1(e)}^{Q_1(e)}}, {I_{v_2(e)}^{Q_2(e)}}}(x_{v_1(e)}, x_{v_2(e)}, 1/t_e)\prod_{h \in T(\gamma)} \frac{\partial \alpha(x_{v(h)})}{\partial x^{I_{v(h)}^{Q(h)}}}.
\end{align}

As a consequence of Proposition \ref{thm:derivatives-k_t}, if we group the terms in $w_\gamma(P_\epsilon^L, I)$ by their powers of $t$, then for
\begin{align}
  w_\gamma(P, I)[\alpha] = \int_{(\epsilon, L)^{|E(\gamma)|}} f_{\gamma, I}(\mathbf{t})[\alpha].
\end{align}
we have, working from (\ref{eq:f_gamma_rn}), that
\begin{align}\label{eq:fgammapi}
  f_{\gamma, I}(\mathbf{t})[\alpha] = \sum_{-O(\gamma) \leq |J| \leq 0 } \mathbf{t}^{J - n/2} \int_{{\RR^n}^{|V(\gamma)|}} e^{-\sum_{e \in E(\gamma)} Q_e/4 t_e} \Phi_J,
\end{align}
where the sum is over multi-indices $J: E(\gamma) \to \ZZ$ and we have defined $\mathbf{t}^J = \prod_{e \in E(\gamma)}t_e^{J(e)}$ and $\mathbf{t}^{-n/2} = \prod_{e \in E(\gamma)}t_e^{-n/2}$ and in the exponential, $Q_e = |x_{v_1(e)} - x_{v_2(e)}|^2$.  Please be careful to distinguish the $Q(h)$ and $Q_e(x)$ in the above formula and subsequent formulas; these are different functions.
Note that $\Phi_J$ is a sum of terms of the form
\begin{align}\label{eq:Phi_J_Diff}
  \prod_{v \in V(\gamma)} D_v \alpha(x_v)
\end{align}
where for each vertex $v$,
\begin{align}
  D_v\alpha = D_{v, 1}\alpha \dots D_{v, l}\alpha
\end{align}
is a product of differential operators applied to $\alpha$.

\subsubsection{Spanning Tree Coordinates}

To evaluate \eqref{eq:fgammapi} make a change of coordinates.  Choose a spanning tree $T$ of $\gamma$.

\begin{defin}
  The coordinates $y_e = x_{v_1(e)} - x_{v_2(e)}$ for $e \in E(T)$ and
  \begin{align}
    w = \sum_{v \in V(\gamma)} x_v
  \end{align}
  form a coordinate system on $\RR^{n|V(\gamma)|}$.  The coordinates $y_e = x_{v_1(e)} - x_{v_2(e)}$ for $e \in E(T)$ and
  \begin{align}
    w = x_{v_0}
  \end{align}
  for some vertex $v_0$ form a coordinate system $\RR^{n|V(\gamma)|}$.  We call each of these coordinate systems \emph{spanning tree coordinates} (in the center of mass and distinguished vertex flavors, respectively).
\end{defin}

The quadratic form $Q_{\mathbf{t}}(x) = \sum_{e \in E(\gamma)} Q_e(x)/4t_e$ can be written in the spanning tree coordinates as $Q_{\mathbf{t}}(w, y)$.

\begin{prop}\label{prop:qindepw}
  The quadratic form $Q_{\mathbf{t}}(w, y)$ is independent of $w$.
\end{prop}

\begin{proof}
  For any edge $e \in E(\gamma)$, let $f^e_1, \dots, f^e_{l(e)}$ be the unique path of edges in $T$ connecting $v_1(e)$ and $v_2(e)$.  Then
  \begin{align}
    x_{v_1(e)} - x_{v_2(e)} = \sum_{i = 1}^{l(e)}(x_{v_1(f^e_i)} - x_{v_2(f^e_i)}) = \sum_{i = 1}^{l(e)}y_{f^e_i}.
  \end{align}
  Therefore,
  \begin{align}
    Q_{\mathbf{t}}(x) &= \sum_{e \in E(\gamma)} Q_e(x)/4t_e\\
         &= \sum_{e \in E(\gamma)}\left| \sum_{i = 1}^{l(e)}y_{f^e_i}\right|^2/4t_e\\
    &= Q_{\mathbf{t}}(y)
  \end{align}
  which clearly does not depend on $w$.
\end{proof}
Let $A(\mathbf{t})$ be the matrix of $Q_{\mathbf{t}}(y) $.  Then $A(\mathbf{t})$ is an $n(|V(\gamma)| - 1)$ by $n(|V(\gamma)| - 1)$ matrix.
\begin{prop}
  The matrix $B(\mathbf{t}) = (4\prod_{e \in E(\gamma)}t_e)A(\mathbf{t})$ has entries that are integer polynomials in $\{t_e\}_{e \in E(\gamma)}$.  Consequently, $P_\gamma(\mathbf{t}) = \det B(\mathbf{t})$ is an integer polynomial in $\{t_e\}_{e \in E(\gamma)}$.
\end{prop}

\begin{prop}\label{prop:APgamma}
  \begin{align}
    \det A(\mathbf{t}) = 4^{-n(|V(\gamma)| - 1)} \mathbf{t}^{-n(|V(\gamma)| - 1)} P_\gamma(\mathbf{t})
  \end{align}
  and thus
  \begin{align}
    A(\mathbf{t})^{-1} = \frac{1}{P_\gamma(\mathbf{t})} C(\mathbf{t}),
  \end{align}
  where $C$ is a matrix with polynomial entries in the $t_e$.
\end{prop}
\begin{proof}
  To prove the second statement, use Cramer's rule
  \begin{align}
    B(\mathbf{t})^{-1} = \frac{1}{\det B(\mathbf{t})} \adj B(\mathbf{t}) = \frac{1}{P_\gamma(\mathbf{t})} \adj B(\mathbf{t})
  \end{align}
  and that $A(\mathbf{t})^{-1} = (4\prod_{e \in E(\gamma)}t_e) B(\mathbf{t})^{-1}$.  So, the statement follows by letting $C(\mathbf{t}) = (4\prod_{e \in E(\gamma)}t_e) \adj B(\mathbf{t})$.
\end{proof}

\subsubsection{Taylor Expansion of $\Phi_J$}

In \eqref{eq:fgammapi}, replace $\Phi_J$ in $f_{\gamma, I}(\mathbf{t})[\alpha]$ with its Taylor polynomial of degree $N'$ in $y$, $\Phi_J^{N'}(w, y) = \sum_{|K| \leq N'} c_{J, K}(w)y^K$, where $N'$ is a non-negative integer to be determined.  This gives
\begin{align}
  f^{N'}_{\gamma, I}(\mathbf{t})[\alpha] &=  \sum_{\substack{|K| \leq N'\\ -O(\gamma) \leq |J| \leq 0} } \mathbf{t}^{J - n/2}\int_{{\RR^n}^{|V(\gamma)|}} e^{-\sum_{e \in E(\gamma)} Q_e(y)/4 t_e} c_{J, K}(w) y^K\, dy dw\\
  &= \sum_{\substack{|K| \leq N', \text{$K$ even}\\ -O(\gamma) \leq |J| \leq 0} }\mathbf{t}^{J - n/2}\mathcal{I}^K(\mathbf{t})  \int_{\RR^n} c_{J, K}(w) \, dw \label{eq:fNprimeIAK}
\end{align}
where
\begin{align}
  \mathcal{I}^K(\mathbf{t}) = \int_{\RR^{n(|V(\gamma)| - 1)}} e^{-Q_{\mathbf{t}}(y)} y^K \, dy.
\end{align}
In fact, $\mathcal{I}^K(\mathbf{t})$ is the square root of a rational function in $\mathbf{t}$ with integer coefficients.

It is clear that $\int_{\RR^n} c_{J, K}(w) \, dw$ is a local functional.  Explicitly, recall that $c_{J, K}(w) = \frac{\partial \Phi_J}{\partial y^K}(0, w)$ and recall from (\ref{eq:Phi_J_Diff}) that
\begin{align}
  \Phi_J = \prod_{v \in V(\gamma)} D_v\alpha(x_v)
\end{align}
where $D_v$ is a $k(v)$-ary multi-differential operator for each $v$.
So,
\begin{align}
  c_{J, K}(w) = \prod_{v \in V(\gamma)} \tilde D_v\alpha(w).
\end{align}
$\tilde D_v$ is a $k(v)$-ary multi-differential operator.  Thus, $\int_{\RR^n} c_{J, K}(w) \, dw$ is a local functional.

The following theorem is a consequence of Wick's theorem:
\begin{thm}\label{thm:fNprime}
  \begin{align}\label{eq:IAK}
    \mathcal{I}^K(\mathbf{t}) = \frac{1}{P_\gamma(\mathbf{t})^{(|K| + 1)/2}} \mathcal{P}^K(\mathbf{t}).
  \end{align}
where $\mathcal{P}^K$ is the square root of a polynomial in $\mathbf{t}$ whose degree depends linearly on $|K|$
\end{thm}
Therefore, we have found the form of $f^{N'}_{\gamma, I}(\mathbf{t})[\alpha]$ from (\ref{eq:fNprimeIAK}) and Theorem \ref{thm:fNprime}.
\begin{cor}\label{cor:ratio_p_q}
There is a collection of local functionals $\Psi_{J, K}(\alpha)$ such that
\begin{align}
  f^{N'}_{\gamma, I}(\mathbf{t})[\alpha] =  \sum_{\substack{|K| \leq N' \\ -O(\gamma) \leq |J| \leq 0} } \frac{\mathcal{P}^K(\mathbf{t})}{\mathcal{Q}^{J,K}(\mathbf{t})} \Psi_{J,K}(\alpha).
\end{align}
where $\mathcal{Q}^{J,K}$ is the square root of a polynomial in $\mathbf{t}$ whose degree depends linearly on $|J|$ and $|K|$.
\end{cor}
More simply, we have found that $f^{N'}_{\gamma, I}(\mathbf{t})[\alpha]$ is the sum of square roots of rational functions in $\mathbf{t}$ times local functionals.
\subsection{Local Counterterms on a Flat Manifold: Error Bounds and Iteration}\label{sec:count-rrn:-error}

\subsubsection{Bounding the Error}
We would like to bound the error
\begin{align}
  E^{N'}_{\gamma, I}(\mathbf{t})[\alpha] = |f_{\gamma, I}(\mathbf{t})[\alpha] - f_{\gamma, I}^{N'}(\mathbf{t})[\alpha]| 
\end{align}

We have the following bounds for $\mathbf{t} \in S_{\text{id}} \cap (0, 1)^{|E(\gamma)|}$, where $S_{\text{id}}$ is the subset of $(0, \infty)^{|E(\gamma)|}$ which was defined in Section \ref{sec:covering_egamma}.
\begin{prop}\label{prop:int_yJ}
  \begin{align}
    \int_{{\RR^n}^{|V(\gamma)|}} e^{-\sum_{e \in E(\gamma)} Q_e/4 t_e} |y^K| \preceq t_{|E(\gamma)|}^{\frac{1}{2}(|K| + n(|V(\gamma)| - 1))} \leq  t_{|E(\gamma)|}^{n(|V(\gamma)| - 1))} 
  \end{align}
\end{prop}

Assume furthermore that $t_{|E(\gamma)|}^R \leq t_1$; that is, $\mathbf{t} \in A_R^{|E(\gamma)|}$ and $\mathbf{t} \in (0, 1)^{|E(\gamma)|}$.
Because $\mathbf{t}^{J - n/2} \leq t_{|E(\gamma)|}^{R(|J| - |E(\gamma)|n/2)} \leq t_{|E(\gamma)|}^{-R(O(\gamma) + |E(\gamma)|n/2)}$,
\begin{align}
  E_{\gamma, I}(\mathbf{t})[\alpha] &\preceq \sum_{-O(\gamma) \leq |J| \leq 0 } \mathbf{t}^{J - n/2} \int_{{\RR^n}^{|V(\gamma)|}} e^{-\sum_{e \in E(\gamma)} Q_e/4 t_e} |\Phi_J - \Phi^{N'}_J|\\
                                                       &\preceq \sum_{|K| = N' + 1} t_{|E(\gamma)|}^{-R(O(\gamma) + |E(\gamma)|n/2)} \int_{{\RR^n}^{|V(\gamma)|}} e^{-\sum_{e \in E(\gamma)} Q_e/4t_{|E(\gamma)|}} e_K(w) |y^K|\\
  &\preceq t_{|E(\gamma)|}^{(N' + 1)/2 + (|V(\gamma)| - 1)/2)} t_{|E(\gamma)|}^{-R(O(\gamma) + |E(\gamma)|n/2)}\sum_{|K| = N' + 1}\int e_K(w) \, dw
\end{align}
using Proposition \ref{prop:int_yJ}.

In the formula above,
\begin{align}
  e_K(w) = \sum_{O(\gamma) \leq |J| \leq 0}\sup_y \left|\frac{\partial \Phi_J}{\partial y^K}(y, w) \right|
\end{align}
and we have
\begin{align}
  \int e_K(w) \, dw \preceq \sum_p \prod_{h \in T(\gamma)} \|\alpha\|_{p_h}
\end{align}
where the summation is over multi-indices $p: T(\gamma) \to \ZZ^{\geq 0}$ such that $\sum_{h \in T(\gamma)} p_h \leq O(\gamma) + N' + 1$, $\|\alpha\|_{p_h}$ is the $C^{p_h}$ norm of $\alpha$.  In the case $M = \RR^n$ this bound is also true if $\|\alpha\|_{p_h}$ is replaced with
\begin{align}
  \|\alpha\|_{(q, p_h)} = \sup_{|J| \leq p_h}\|(1 + |\alpha|)^q \partial_J\alpha \|,
\end{align}
a Schwartz seminorm, for some sufficiently high power $q$.

In conclusion, we have shown that
\begin{thm}\label{thm:f_minus_fNprime}
  \begin{align}
    E^{N'}_{\gamma, I}(\mathbf{t})[\alpha] \preceq \left(\sum_p \prod_{h \in T(\gamma)} \|\alpha\|_{p_h} \right)t_{|E(\gamma)|}^{(N' + 1)/2 + n(|V(\gamma)| - 1)/2 - R(O(\gamma) + n|E(\gamma)|/2)}
  \end{align}
  where the summation is over multi-indices $p: T(\gamma) \to \ZZ^{\geq 0}$ such that $\sum_{h \in T(\gamma)} p_h \leq O(\gamma) + N' + 1$, $\|\alpha\|_{p_h}$ is the $C^{p_h}$ norm of $\alpha$.  When working globally on $M = \RR^n$, $\|\alpha\|_{p_h}$ is replaced with $\|\alpha\|_{(q, p_h)}$, a Schwartz seminorm, for some sufficiently high power $q$.
\end{thm}

We will need below that by Proposition \ref{prop:int_yJ},
\begin{prop}\label{prop:intAKbounds}
  \begin{align}
    |\mathcal{I}^K(\mathbf{t})| \leq C t_{|E(\gamma)|}^{(|K| + n(|V(\gamma)| - 1))/2} \leq C t_{|E(\gamma)|}^{n(|V(\gamma)| - 1)/2} 
  \end{align}
  and consequently,
  \begin{align}\label{eqfNprimebound}
    |f^{N'}_{\gamma, I}(\mathbf{t})[\alpha]| &\preceq\left(\sum_p \prod_{h \in T(\gamma)} \|\alpha\|_{p_h} \right)t_{|E(\gamma)|}^{n(|V(\gamma)| - 1)/2 - R(O(\gamma) + n|E(\gamma)|/2)}
  \end{align}
  where the summation is over multi-indices $p: T(\gamma) \to \ZZ^{\geq 0}$ such that $\sum_{h \in T(\gamma)} p_h \leq O(\gamma) + N'$ and $\|\alpha\|_{p_h}$ is the $C^{p_h}$ norm of $\alpha$.  When working globally on $M = \RR^n$, $\|\alpha\|_{p_h}$ is replaced with $\|\alpha\|_{(q, p_h)}$, a Schwartz seminorm, for some sufficiently high power $q$.
\end{prop}

\subsubsection{Inductive Construction of the Counterterms}
\label{sec:induct-constr-count}
In this section, we shall construct the counterterms $w_{\gamma}^{\text{ct}}(P_{\epsilon}^L, I)$ for the Feynman weight
\begin{align}
  w_{\gamma}(P_{\epsilon}^L, I) = \int_{(\epsilon, L)^{|E(\gamma)|}} f_{\gamma, I}(\mathbf{t})[\alpha].
\end{align}

We shall use the results of Section \ref{sec:covering_egamma}, in particular Corollary \ref{cor:refined_cover}, which gives a finite cover of $S_{\text{id}}$ by sets of the form $\overline{E}^{I^{(1)}}_{R} \cap \overline{E}^{I^{(2)}}_{R}\dots \cap \overline{E}^{I^{(p)}}_{R}$, for $p \leq |E(\gamma)|$.  The structure of the multi-indices $I^{(1)}, \dots, I^{(p)}$ is given in Corollary \ref{cor:refined_cover} and the sets $E_R^I$ are defined in Definition \ref{def:E_RI}.

Also we shall need a slightly generalized form of Theorem \ref{thm:ERcontainedinAi} which states that for $E_R^I$ where $I$ is a sequence of the form $1 \leq i_0 < i_1 < \dots < i_m \leq |E(\gamma)|$ with $m \leq |E(\gamma)| - 1$, we have $E_R^I \subseteq A^{i_0, i_m}_{R^{2^m}}$, where
\begin{align}
  A^{i_0, i_m}_{R^{2^m}} = \{t_1 < t_2 < \dots < t_{|E(\gamma)|} : \text{$t_{i_m} < t_{i_m + 1}^{R^{2^m}}$ and $t^{R^{2^m}}_{i_m} < t_{i_0}$}\}.
\end{align}
Recall that Theorem \ref{thm:ERcontainedinAi} stated this result for the special case where $i_0 = 1$.

Note that the procedure we describe below should be carried out separately in $S_{\sigma}$ for each permutation $\sigma$.  However, we work with $S_{\text{id}}$ without loss of generality for notational simplicity.

\begin{thm}\label{thm:inductive_step}
  For any sequence $I^{(1)}, \dots, I^{(p)}$ as in Corollary \ref{cor:refined_cover}, for nonnegative integers $N_1', \dots, N_p'$, we can construct $f_{\gamma, I}^{N_1', \dots, N_p'}(\mathbf{t})[\alpha]$ by iterative Taylor expansion of $f_{\gamma, I}(\mathbf{t})[\alpha]$ so that for all $\mathbf{t} \in \overline{E}^{I^{(1)}}_{R} \cap \overline{E}^{I^{(2)}}_{R}\dots \cap \overline{E}^{I^{(p)}}_{R}$,
  \begin{align}
    |f_{\gamma, I}(\mathbf{t})[\alpha] - f^{N_1', \dots, N_p'}_{\gamma, I}(\mathbf{t})[\alpha]| \preceq \|\alpha\|_l^{|T(\gamma)|} \sum_{i = 1}^p t_{|E(\gamma)|}^{d_i},
  \end{align}
  where $l$ is some positive integer, and where $d_i = d_i(N_1', \dots, N_i')$ increases linearly in $N_i'$ for $N'_1, \dots, N'_{i - 1}$ fixed and sufficiently large.   When working globally on $M = \RR^n$, $\|\alpha\|_l$, the $C^l$ seminorm, is replaced with $\|\alpha\|_{(q, l)}$, a Schwartz seminorm, for some sufficiently high power $q$.
\end{thm}

\begin{proof}
For simplicity we shall assume that $p = 2$.  The argument is essentially the same in general, but we would like to avoid some details at the end that would make the notation even more complex.  Let $i^{(1)} = i^{(1)}_{m^{(1)}}$ and let $R_1 = R^{s_{m^{(1)} + 1}}$ and $R_2 = R^{s_{m^{(2)} + 1}}$, where recall that $s_i := 2^{i - 1}$, so that we are working within
  \begin{align}
    \overline{E}_R^{I^{(1)}} \cap \overline{E}_R^{I^{(2)}} &\subseteq \overline{A}^{1, i^{(1)}}_{R_1} \cap \overline{A}^{i^{(1)} + 1, k}_{R_2}\\
                                                           &= \{\mathbf{t} \in S_{\text{id}} : \text{$t^{R_1}_{i^{(1)}} \leq t_1$ and $t_{i^{(1)}} \leq t^{R_1}_{i^{(1)} + 1}$ and $t^{R_2}_{|E(\gamma)|} \leq t_{i^{(1)} + 1}$} \}
  \end{align}
  where the inclusion follows from the generalized form of Theorem \ref{thm:ERcontainedinAi}.  Note that $i^{(2)}_{m^{(2)}} = |E(\gamma)|$.  

   The collection of edges $e_1, \dots, e_{i^{(1)}}$ determines a subgraph of $\gamma$, which we denote by $\gamma'$.  Note that we include in $\gamma'$ all the half edges coming from edges in $\gamma$ that are incident on vertices of $\gamma'$.  In the same way, the remaining vertices $e_{i^{(1)} + 1}, \dots, e_{|E(\gamma)|}$ induce a subgraph that we denote by $\gamma \setminus \gamma'$.  The integral in the formula for $f_{\gamma, I}(\mathbf{t})[\alpha]$ is over $\RR^{n|V(\gamma)|}$ and we can order the integration so that we integrate first with respect to the vertices in $V(\gamma')$ and then by the vertices of $V(\gamma \setminus \gamma')$.  Our argument proceeds by finding bounds that hold on $\gamma'$ and its complement $\gamma \setminus \gamma'$.   

 Firstly, we define $g_{\gamma', \gamma, I}(\mathbf{t})[\alpha]$ to be the integral over $\RR^{n|V(\gamma')|}$ of the product of
  \begin{align}\label{eq:inner_int}
    \prod_{v \in V(\gamma')} a^{I_v^1, \dots, I_v^{k(v)}}(x_v) \prod_{e \in E(\gamma')} \frac{\partial K_{t_e}(x_{v_1(e)}, x_{v_2(e)})}{\partial x^{I_{v_1(e)}^{Q_1(e)}}\partial x^{I_{v_2(e)}^{Q_2(e)}}} 
  \end{align}
  and
  \begin{align}
    \prod_{e \in E(\gamma', \gamma)} \frac{\partial K_{t_e}(x_{v_1(e)}, x_{v_2(e)})}{\partial x^{I_{v_1(e)}^{Q_1(e)}}} \prod_{h \in T(\gamma', \gamma)} \frac{\partial \alpha(x_{v(h)})}{\partial x^{I_{v(h)}^{Q(h)}}}
  \end{align}
  Here we have defined $E(\gamma', \gamma)$ be the set of all edges in $\gamma$ for which one half edge (by convention the first half edge) making up the edge is a tail in $\gamma'$, and we have defined $T(\gamma', \gamma)$ to be the collection of tails of $\gamma'$ that are also tails of $\gamma$.  By definition, $g_{\gamma', \gamma, I}(\mathbf{t})[\alpha]$ is a function of $\mathbf{t} = \{t_e\}_{e \in E(\gamma') \cup E(\gamma', \gamma)}$ and $x_v$ for $v \in V(\gamma', \gamma)$, the set of second vertices of edges in $E(\gamma', \gamma)$.

  Secondly, we define for any $g$, a function of the variables $\{x_v\}_{v \in V(\gamma', \gamma)}$, $f_{\gamma \setminus \gamma', I, g}(\mathbf{t})$ to be the integral over $\RR^{n|V(\gamma) \setminus V(\gamma')|}$ of the product of 
  \begin{align}
    \prod_{v \in V(\gamma \setminus \gamma')} a^{I_v^1, \dots, I_v^{k(v)}}(x_v) \prod_{e \in E(\gamma \setminus \gamma')} \frac{\partial K_{t_e}(x_{v_1(e)}, x_{v_2(e)})}{\partial x^{I_{v_1(e)}^{Q_1(e)}}\partial x^{I_{v_2(e)}^{Q_2(e)}}}
  \end{align}
  and
  \begin{align}
    \frac{\partial g}{\prod_{h \in S(\gamma \setminus \gamma', \gamma)} \partial x^{I_{v(h)}^{Q(h)}}} \prod_{h \in T(\gamma \setminus \gamma', \gamma) } \frac{\partial \alpha(x_{v(h)})}{\partial x^{I_{v(h)}^{Q(h)}}}.
  \end{align}
  Here we have defined $S(\gamma \setminus \gamma', \gamma)$ to be the set of all tails in $\gamma \setminus \gamma'$ that are not tails in $\gamma$, and we have defined $T(\gamma \setminus \gamma', \gamma)$ to be the set of tails in $\gamma \setminus \gamma'$ that are also tails in $\gamma$.  We shall call the elements of $S(\gamma \setminus \gamma', \gamma)$ special tails.  By definition $f_{\gamma, I}(\mathbf{t})[\alpha] = f_{\gamma/\gamma', I, g_{\gamma', \gamma, I}}(\mathbf{t})$.

  Using the same procedure which led to Theorem \ref{thm:f_minus_fNprime}, we have for
  \begin{align}
    E^{N_1'}_{\gamma', \gamma, I}(\mathbf{t})[\alpha] = \left|g_{\gamma', \gamma, I}(\mathbf{t})[\alpha] - g_{\gamma', \gamma, I}^{N_{1}'}(\mathbf{t})[\alpha]\right|
  \end{align}
  that
  \begin{align}
    E^{N_1'}_{\gamma', \gamma, I}(\mathbf{t})[\alpha] \preceq \left(\sum_p \prod_{h \in T(\gamma', \gamma)} \|\alpha\|_{p_h} \prod_{e \in E(\gamma', \gamma)} \|K_{t_e}\|_{p_{h(e)}} \right)t_{i^{(1)}}^{N_1'/2 + C(\gamma', n, R_1)}
  \end{align}
  where $h(e)$ is the second half edge of $e \in E(\gamma', \gamma)$ and
  \begin{align}
    C(\gamma', n, R_1) = {1/2 + (|V(\gamma')| - 1)n/2 - R_1\left(O(\gamma') + |E(\gamma')|n/2 \right)}.
  \end{align}
  and where the sum is over all multi-indices $p:T(\gamma') \to \ZZ^{\geq 0}$ such that $\sum_{h \in T(\gamma')}p_h \leq O(\gamma') + N_1' + 1$.

  But for $e \in E(\gamma', \gamma)$,
  \begin{align}
    \|K_{t_e}\|_{p_{h(e)}} \preceq  t_e^{-n/2 - p_{h(e)}} \leq t_{i^{(1)} + 1}^{-n/2 - p_{h(e)}}
  \end{align}
  for some constant $C$, and thus
  \begin{align}
    \prod_{e \in E(\gamma', \gamma)} \|K_{t_e}\|_{p_{h(e)}} & \preceq t_{i^{(1)} + 1}^{-n|E(\gamma', \gamma)|/2 - \sum_{e \in E(\gamma', \gamma)} p_{h(e)}}\\
                                                            & \leq t_{i^{(1)} + 1}^{-n|E(\gamma)|/2 - O(\gamma') - N_1' - 1}
  \end{align}
  So finally, we have
  \begin{align}
    E^{N_1'}_{\gamma', \gamma, I}(\mathbf{t})[\alpha] \preceq \left(\sum_p \prod_{h \in T(\gamma', \gamma)} \|\alpha\|_{p_h} \right)t_{i^{(1)} + 1 }^{(R_1/2 - 1)N_1'  + R_1 C(\gamma', n, R_1) - n|E(\gamma)|/2 - O(\gamma') - 1}.
  \end{align}
  As long as $R_1 > 2$, $E^{N_1'}_{\gamma', \gamma, I}(\mathbf{t})[\alpha] = \left|g_{\gamma', \gamma, I}(\mathbf{t})[\alpha] - g_{\gamma', \gamma, I}^{N_{1}'}(\mathbf{t})[\alpha]\right|$ and consequently
  \begin{align}
    \left|f_{\gamma \setminus \gamma', I, g_{\gamma', \gamma, I}}(\mathbf{t})[\alpha] - f_{\gamma \setminus \gamma', I, g^{N_1'}_{\gamma', \gamma, I}}(\mathbf{t})[\alpha]\right|
  \end{align}
  will be bounded by a power of $t_{|E(\gamma)|}$, which we denote $d_1(N_1')$, that grows linearly with $N_1'$.

  To finish the argument, we need to find a bound for
  \begin{align}
    E^{N_1', N_2'}_{\gamma', \gamma, I}(\mathbf{t})[\alpha] = \left|f_{\gamma \setminus \gamma', I, g^{N_1'}_{\gamma', \gamma, I}}(\mathbf{t})[\alpha] - f^{N_2'}_{\gamma \setminus \gamma', I, g^{N_1'}_{\gamma', \gamma, I}}(\mathbf{t})[\alpha]\right|
  \end{align}
  As in the proof of Theorem \ref{thm:f_minus_fNprime}, we have $\left|f_{\gamma \setminus \gamma', I, g}(\mathbf{t})[\alpha] - f^{N_2'}_{\gamma \setminus \gamma', I, g}(\mathbf{t})[\alpha]\right| \preceq$
  \begin{align}
    \left(\sum_p\prod_{h \in T(\gamma \setminus \gamma', \gamma)} \|\alpha\|_{p_h}\right) \|g\|_l\, t_{|E(\gamma)|}^{(N_2' + 1)/2 + n(|V(\gamma \setminus \gamma')| - 1)/2 - R_2(O(\gamma \setminus \gamma') + n|E(\gamma \setminus \gamma')|/2)}\\
  \end{align}
  By reasoning similar to Proposition \ref{prop:intAKbounds}, we have $\|g^{N_1'}_{\gamma', \gamma, I}(\mathbf{t})[\alpha]\|_l \preceq$ 
  \begin{align}
    &\left(\sum_p \prod_{h \in T(\gamma', \gamma)} \|\alpha\|_{p_h} \prod_{e \in E(\gamma', \gamma)} \|K_{t_e}\|_{p_{h(e)}} \right)t_{i^{(1)}}^{n(|V(\gamma')| - 1)/2 - R_1(O(\gamma') + n|E(\gamma')|/2)}\\
    \preceq & \left(\sum_p \prod_{h \in T(\gamma', \gamma)} \|\alpha\|_{p_h} \right)t_{i^{(1)} + 1}^{C_1(N_1', \gamma', \gamma, n, R_1) - l}
  \end{align}
for some constant $C_1(N_1', \gamma', \gamma, n, R_1)$.
Therefore, we have
\begin{align}
  E^{N_1', N_2'}_{\gamma', \gamma, I}(\mathbf{t})[\alpha] \preceq \left(\sum_p\prod_{h \in T(\gamma)} \|\alpha\|_{p_h}\right) t_{|E(\gamma)|}^{N_2'/2 + C_2(N_1', \gamma', \gamma, n, R_1, R_2)}
\end{align}
for some constant $C_2(N_1', \gamma', \gamma, n, R_1, R_2)$.

But
\begin{align}
  \left(\sum_p\prod_{h \in T(\gamma)} \|\alpha\|_{p_h}\right) \preceq \|\alpha\|_l^{|T(\gamma)|}
\end{align}
for some positive integer $l$.  In conclusion, using the triangle inequality, we can bound
  \begin{align}
    |f_{\gamma/\gamma', I, g_{\gamma', \gamma, I}}(\mathbf{t})[\alpha] - f^{N'_2}_{\gamma/\gamma', I, g^{N_1}_{\gamma, \gamma', I}}(\mathbf{t})[\alpha]| \preceq\|\alpha\|_l^{{|T(\gamma)|}}\left(t_{|E(\gamma)|}^{d_1(N'_1)} + t_{|E(\gamma)|}^{d_2(N_1', N_2')}\right)
  \end{align}
  where by $d_1(N_1')$ grows linearly in $N_1'$ and $d_2(N_1', N_2')$ grows linearly in $N_2'$ for $N_1'$ fixed.
\end{proof}

\subsection{Local Counterterms on a Flat Manifold with Boundary}
\label{sec:count_upp_half}
As we previously wrote, in the case of the Euclidean half space $\HH^n$, the Dirichlet heat kernel is given by
\begin{align}
  K_t(x_1, x_2) = (4\pi t)^{-n/2}[e^{-|x_1 - x_2|^2/4t} - e^{-|x_1 - x_2^*|^2/4t}]
\end{align}
and the Neumann heat kernel is given by
\begin{align}
  K_t(x_1, x_2) = (4\pi t)^{-n/2}[e^{-|x_1 - x_2|^2/4t} + e^{-|x_1 - x_2^*|^2/4t}]
\end{align}
where $x_2^{*}$ is the reflection through the boundary.  In coordinates near the boundary of a flat manifold with boundary we shall also assume that the heat kernel is of this form.  The justification for this will be given in Section \ref{sec:count-comp-manif-bdry}.

We give a combined analysis of the procedure for both the Dirichlet and Neumann heat kernel by absorbing any signs into constants in the formulas, and, as we did in the case of a flat manifold without boundary, a combined analysis of the global $\HH^n$ case and the local case.

Slightly modifying our procedure in Section \ref{sec:count_rn}, we form 
\begin{align}
  w_\gamma(P_\epsilon^L, I)[\alpha] = \sum_{\beta} \int_{(\epsilon, L)^{|E(\gamma)|}} f_{\gamma, I, \beta}(\mathbf{t})[\alpha]
\end{align}
where $\beta$ ranges over all functions $E(\gamma) \to \{-1, 1\}$ and
\begin{align}\label{fgammaibeta}
  f_{\gamma, I, \beta}(\mathbf{t})[\alpha] = \sum_{-O(\gamma) \leq |J| \leq 0 } \mathbf{t}^{J - n/2}  \int_{{\HH^n}^{|V(\gamma)|}} e^{-\sum_{e \in E(\gamma)} Q^{(\beta_e)}_e/4 t_e}  \Phi_{J, \beta}
\end{align}
where $Q_e^{(1)} = |x_{v_1(e)} - x_{v_2(e)}|^2$ and $Q_e^{(-1)} = |x_{v_1(e)} - x_{v_2(e)}^*|^2$.  We wish to apply Wick's theorem after taking the Taylor expansion of $\Phi_{J, \beta}$.  

\subsubsection{One Possible Coordinate System on $\HH^{n|V(\gamma)|}$}\label{sec:one-poss-coord}
Using the decomposition $\HH^{n} = \RR^{(n - 1)} \times \RR_{\geq 0}$ define $x_v = (\overline{x}_v, x_{v, n})$.  For the choice of a spanning tree $T$, the coordinates
\begin{align}
  \overline{y}_e &= \overline{x}_{v_1(e)} - \overline{x}_{v_2(e)}\\
  y_{e, n} &= x_{v_1(e), n} - \beta(e) x_{v_2(e), n}
\end{align}
for $e \in E(T)$ and
\begin{align}
  \overline{w} &= \sum_{v \in V(\gamma)} x_v\\
  w_n &= x_{v_0, n},
\end{align}
for some choice of vertex $v_0$, form a coordinate system on $\HH^{n|V(\gamma)|}$.  Note that as an alternative, one could define $\overline{w} = \overline{x}_{v_0}$.  We write $y_e = (\overline{y}_e, y_{e, n})$ and $w = (\overline{w}, w_n)$.

The quadratic form $Q_{\mathbf{t}}(x) = \sum_{e \in E(\gamma)} Q_e(x)/4t_e$ is equal to
\begin{align}
  Q_{\mathbf{t}}(\overline{x}) = \sum_{e \in E(\gamma)} |\overline{x}_{v_1(e)} - \overline{x}_{v_2(e)}|/4t_e 
\end{align}
plus the part depending on the variables $x_{v, n}$
\begin{align}
  Q^{(\beta)}_{\mathbf{t}}(x_n) = \sum_{e \in \beta^{-1}(1)}|x_{v_1(e), n} - x_{v_1(e), n}|^2/4t_e + \sum_{e \in \beta^{-1}(-1)}|x_{v_1(e), n} + x_{v_1(e), n}|^2/4t_e
\end{align}

We would like to express $Q_{\mathbf{t}}(\overline{x})$ and $Q^{(\beta)}_{\mathbf{t}}(x_n)$ in terms of the variables $y_{e, n}$ and $w_n$.  The quadratic form $Q_{\mathbf{t}}(\overline{x})$, as in the proof of Proposition \ref{prop:qindepw}, is given by
\begin{align}
  Q_{\mathbf{t}}(\overline y) = \sum_{e \in E(\gamma)}\left| \sum_{i = 1}^{l(e)}\overline y_{f^e_i}\right|^2/4t_e,
\end{align}
As before, for any edge $e \in E(\gamma)$, we define $f^e_1, \dots, f^e_{l(e)}$ to be the unique path of edges in $T$ connecting $v_1(e)$ and $v_2(e)$ 

For the direction normal to the boundary, suppose that $f^v_1, \dots, f^v_{l(v)}$ is the sequence of edges in $T$ forming a path between $v_0$ and some vertex $v$.  We have
\begin{align}
  \sum_{i = 1}^{l(v)}\left(\prod_{j = 1}^i\beta(f_j^v)\right)y_{f^v_i, n} &= \sum_{i = 1}^{l(v)}\left(\prod_{j = 1}^i\beta(f_j^v)\right)(x_{v_2(f_i^v), n} - \beta(f_i^v) x_{v_1(f_i^v), n})\\
  &= \left(\prod_{j = 1}^{l(v)}\beta(f_j^v)\right)x_{v, n} - x_{v_0, n}
\end{align}
Therefore, $|x_{v_1(e), n} - \beta(e) x_{v_2(e), n}|^2$ can depend on $w_n$.

The quadratic form $Q_{\mathbf{t}}^{(\beta)}(x_n)$ decomposes into a sum of two terms
\begin{align}
  \widetilde{Q}_{\mathbf{t}}^{(\beta)}(y_n, w_n) + Q^{(\beta)}_{\mathbf{t}}(w_n), 
\end{align}
where $\widetilde{Q}_{\mathbf{t}}^{(\beta)}(y_n, w_n)$ is a nondegenerate quadratic form in $y_n$ that is possibly inhomogeneous (with no constant part) and $Q^{(\beta)}_{\mathbf{t}}(w_n)$ is a homogeneous quadratic form in $w_n$.  We say that $\beta$ is in the first case if $Q_{\mathbf{t}}^{(\beta)}(x_n)$ only depends on $y_n$ and $\beta$ is in the second case if $Q_{\mathbf{t}}^{(\beta)}(x_n)$ depends on both $y_n$ and $w_n$.

From the fact that $x_{v, n} \in \RR_{\geq 0}$ for all $v \in V(\gamma)$ and therefore $w_n \in \RR_{\geq 0}$, it follows that $y_{e, n}$ ranges over some polyhedron $P(w_n) \subset \RR^{|V(\gamma)| - 1}$.

\subsubsection{A Second Possible Coordinate System on $\HH^{n|V(\gamma)|}$}
As in the case of $\RR^n$ for the choice of a spanning tree $T$, we can use the coordinates $y_e = x_{v_1(e)} - x_{v_2(e)}$ for $e \in E(T)$ and
\begin{align}
  w = \sum_{v \in V(\gamma)} x_v
\end{align}
form a coordinate system on $\RR^{n|V(\gamma)|}$.

Using the decomposition $\HH^{n} = \RR^{(n - 1)} \times \RR_{\geq 0}$ define $x_v = (\overline{x}_v, x_{v, n})$, $y_e = (\overline{y}_e, y_{e, n})$ and $w = (\overline{w}, w_n)$.  The quadratic form $Q_{\mathbf{t}}(x) = \sum_{e \in E(\gamma)} Q_e(x)/4t_e$ is equal to
\begin{align}
  Q_{\mathbf{t}}(\overline{x}) = \sum_{e \in E(\gamma)} |\overline{x}_{v_1(e)} - \overline{x}_{v_2(e)}|/4t_e 
\end{align}
plus the part depending on the variables $x_{v, n}$
\begin{align}
  Q^{(\beta)}_{\mathbf{t}}(x_n) = \sum_{e \in \beta^{-1}(1)}|x_{v_1(e), n} - x_{v_1(e), n}|^2/4t_e + \sum_{e \in \beta^{-1}(-1)}|x_{v_1(e), n} + x_{v_1(e), n}|^2/4t_e
\end{align}

We would now like to express $Q^{(\beta)}_{\mathbf{t}}(x_n)$ in terms of the variables $y_{e, n}$ and $w_n$.
While
\begin{align}
  \sum_{e \in \beta^{-1}(1)}|x_{v_1(e), n} - x_{v_1(e), n}|^2/4t_e &= \sum_{e \in \beta^{-1}(1)}\left|\sum_{i = 1}^{l(e)}y_{f^e_i, n}\right|^2\bigg/4t_e,
\end{align}
unfortunately $\sum_{e \in \beta^{-1}(-1)}|x_{v_1(e), n} + x_{v_1(e), n}|^2/4t_e$ has the more complicated expression
\begin{align}
   \frac{1}{|V(\gamma)|^2}\sum_{e \in \beta^{-1}(-1)} \left| 2w_n + \sum_{v' \neq v_1(e)}\sum_{i = 1}^{l_1(v')}y_{f^{v', 1}_i, n} + \sum_{v' \neq v_2(e)}\sum_{i = 1}^{l_2(v')}y_{f^{v', 2}_i, n} \right|^2\bigg/4t_e
\end{align}
where $f^{v', 1}_1, \dots, f^{v', 1}_{l_1(v')}$ is the unique path of edges in $T$ connecting $v_1(e)$ and $v'$ and $f^{v', 2}_1, \dots, f^{v', 2}_{l_2(v')}$ is the unique path of edges in $T$ connecting $v_2(e)$ and $v'$.  Therefore the quadratic form $Q_{\mathbf{t}}^{(\beta)}(x_n)$ decomposes into a sum of two terms
\begin{align}
  \widetilde{Q}_{\mathbf{t}}^{(\beta)}(y_n, w_n) + Q^{(\beta)}_{\mathbf{t}}(w_n), 
\end{align}
where $\widetilde{Q}_{\mathbf{t}}^{(\beta)}(y_n, w_n)$ is a quadratic form in $y_n$ that is possibly inhomogeneous and possibly degenerate in $y_n$ and lastly
\begin{align}
  Q^{(\beta)}_{\mathbf{t}}(w_n) = \frac{1}{|V(\gamma)|^2}\left(\sum_{e \in \beta^{-1}(-1)} t_e^{-1}\right)w_n^2.
\end{align}
From the fact that $x_{v, n} \in \RR_{\geq 0}$ for all $v \in V(\gamma)$ and therefore $w_n \in \RR_{\geq 0}$, it follows that $y_{e, n}$ ranges over some bounded polytope $P(w_n) \subset \RR^{|V(\gamma)| - 1}$.

The coordinate system in this section, compared to the one in Section \ref{sec:one-poss-coord}, has the benefit that $P(w_n)$ is bounded, but the drawback that $\widetilde{Q}_{\mathbf{t}}^{(\beta)}(y_n, w_n)$ can be degenerate in $y_n$.  It is still possible to carry out the construction of the counterterms using the coordinate system in this section analogously to the construction we give below using the coordinate system of Section \ref{sec:one-poss-coord}.  There is even the added benefit that if the interaction is translation invariant normal to the boundary, the counterterms will be as well.  The coordinate system of Section \ref{sec:one-poss-coord} tends to be easier to work with because of the non-degeneracy of $\widetilde{Q}_{\mathbf{t}}^{(\beta)}(y_n, w_n)$, so we focus on in exclusively in the next sections.

\subsubsection{Taylor Expansion of $\Phi_J$}
Beginning from \eqref{fgammaibeta}, replace $\Phi_{J, \beta}$ in $f_{\gamma, I, \beta}(\mathbf{t})[\alpha]$ with its Taylor polynomial of degree $N'$ in $y$,
\begin{align}
  \Phi_{J, \beta}^{N'}(w, y) = \sum_{|\overline{K}| + |K_n| \leq N'} c_{J, \overline{K}, K_n}(w)\overline{y}^{\overline{K}}y_n^{K_n},
\end{align}
where $N'$ is a non-negative integer to be determined.  This gives
\begin{align}
  f^{N'}_{\gamma, I, \beta}(\mathbf{t})[\alpha] &=  \sum_{\substack{|\overline{K}| + |K_n| \leq N'\\ -O(\gamma) \leq |J| \leq 0} }\mathbf{t}^{J - n/2}\overline{\mathcal{I}}^{\overline{K}}(\mathbf{t})  \int_{\HH^n} e^{-Q_{\mathbf{t}}^{(\beta)}(w_n)}\mathcal{I}_{\beta}^{K_n}(w_n, \mathbf{t}) c_{J, \overline{K}, K_n}(w) \, dw 
\end{align}
where
\begin{align}
  \overline{\mathcal{I}}^{\overline{K}}(\mathbf{t}) = \int_{\RR^{(n - 1)(|V(\gamma)| - 1)}} e^{-Q_{\mathbf{t}}(\overline{y})} \overline{y}^{\overline{K}} \, d\overline{y}.
\end{align}
and
\begin{align}
  \mathcal{I}_{\beta}^{K_n}(w_n, \mathbf{t}) = \int_{P(w_n)}e^{-\widetilde{Q}_{\mathbf{t}}^{(\beta)}(y_n, w_n)} y_n^{K_n} \, dy_n
\end{align}

The integral over $\overline y$ defining $\overline{\mathcal{I}}^{\overline{K}}(\mathbf{t})$ gives an answer like that of Corollary \ref{cor:ratio_p_q}, but with the dimension $n$ replaced by $n - 1$ in all the formulas. 
The problem is that $f_{\gamma, I, \beta}^{N'}(\mathbf{t})[\alpha]$ is not quite the product of a function of $\mathbf{t}$ and a local functional because of the factor $e^{-Q_{\mathbf{t}}^{(\beta)}(w_n)}\mathcal{I}_{\beta}^{K_n}(w_n, \mathbf{t})$ that depends on both $w_n$ and $\mathbf{t}$.  To remedy this, we would like to introduce a Taylor expansion in $w_n$.

We split into two cases.  If $\beta$ is in the first case, $\widetilde{Q}_{\mathbf{t}}^{(\beta)}(y_n, w_n)$ is homogeneous and independent of $w_n$ and $\widetilde{Q}_{\mathbf{t}}^{(\beta)}(w_n) = 0$.  We write
\begin{align}
  \int_{\HH^n}\mathcal{I}_{\beta}^{K_n}(w_n, \mathbf{t}) c_{J, \overline{K}, K_n}(w) \, dw &= \mathcal{I}_{\beta}^{K_n}(\infty, \mathbf{t}) \int_{\HH^n}c_{J, \overline{K}, K_n}(w) \, dw \\
  &+ \int_{\HH^n}(\mathcal{I}_{\beta}^{K_n}(w_n, \mathbf{t}) - \mathcal{I}_{\beta}^{K_n}(\infty, \mathbf{t}))c_{J, \overline{K}, K_n}(w) \, dw
\end{align}
Note that $\mathcal{I}_{\beta}^{K_n}(\infty, \mathbf{t})$ exists because $\widetilde{Q}_{\mathbf{t}}^{(\beta)}(y_n, w_n)$ is non-degenerate.
The first integral is already a local functional of $\alpha$ on $\HH^n$ and carries on in the construction.  For the second integral, we need to take the Taylor expansion of $c_{J, \overline{K}, K_n}(w)$ in $w_n$ which we write as
\begin{align}
  c^{m'}_{J, \overline{K}, K_n}(w) = \sum_{i = 0}^{m'} d_{J, \overline{K}, K_n, i}(\overline{w})w_n^i.
\end{align}
That is, we define
\begin{align}
  f^{N', m'}_{\gamma, I, \beta}(\mathbf{t})[\alpha] &= \sum_{\substack{|\overline{K}| + |K_n| \leq N'\\ -O(\gamma) \leq |J| \leq 0} }\mathbf{t}^{J - n/2}\overline{\mathcal{I}}^{\overline{K}}(\mathbf{t}) \mathcal{I}_{\beta}^{K_n}(\infty, \mathbf{t}) \int_{\HH^n} c_{J, \overline{K}, K_n}(w) \, dw\\
  &+\sum_{\substack{|\overline{K}| + |K_n| \leq N'\\ -O(\gamma) \leq |J| \leq 0} }\sum_{i = 0}^{m'}\mathbf{t}^{J - n/2}\overline{\mathcal{I}}^{\overline{K}}(\mathbf{t})  \mathcal{J}^{K_n, i}_{\beta} (\mathbf{t})\int_{\RR^{n - 1}}d_{J, \overline{K}, K_n, i}(\overline{w}) \, d\overline{w}
\end{align}
where
\begin{align}
  \mathcal{J}^{K_n, i}_{\beta}(\mathbf{t}) = \int_{\RR_{\geq 0}} (\mathcal{I}_{\beta}^{K_n}(w_n, \mathbf{t}) - \mathcal{I}_{\beta}^{K_n}(\infty, \mathbf{t})) w_n^i dw_n
\end{align}
and now $\int_{\RR^{n - 1}}d_{J, \overline{K}, K_n, i}(\overline{w}) \, d\overline{w}$ is a local functional of $\alpha$ on $\RR^{n - 1}$.

If $\beta$ is in the second case, we again write
\begin{align}
  c^{m'}_{J, \overline{K}, K_n}(w) = \sum_{i = 0}^{m'} d_{J, \overline{K}, K_n, i}(\overline{w})w_n^i.
\end{align}
and define
\begin{align}
    f^{N', m'}_{\gamma, I, \beta}(\mathbf{t})[\alpha] &=  \sum_{\substack{|\overline{K}| + |K_n| \leq N'\\ -O(\gamma) \leq |J| \leq 0} } \sum_{i = 0}^{m'} t^{J - n/2}\overline{\mathcal{I}}^{\overline{K}}(\mathbf{t}) \mathcal{K}^{K_n, i}_{\beta}(\mathbf{t}) \int_{\RR^n} d_{J, \overline{K}, K_n, i}(\overline{w}) \, d\overline{w} 
\end{align}
where
\begin{align}
  \mathcal{K}^{K_n, i}_{\beta}(\mathbf{t}) = \int_{\RR_{\geq 0}} e^{-Q_{\mathbf{t}}^{(\beta)}(w_n)}\mathcal{I}_{\beta}^{K_n}(w_n, \mathbf{t}) w_n^i \, dw_n 
\end{align}
and again $\int_{\RR^{n - 1}}d_{J, \overline{K}, K_n, i}(\overline{w}) \, d\overline{w}$ is a local functional of $\alpha$ on $\RR^{n - 1}$, as desired.

\subsubsection{Bounding the Error}
We would like to bound the error
\begin{align}
  E^{N'}_{\gamma, I, \beta}(\mathbf{t})[\alpha] = |f_{\gamma, I, \beta}(\mathbf{t})[\alpha] - f_{\gamma, I, \beta}^{N'}(\mathbf{t})[\alpha]| 
\end{align}
Fix an ordering of the set of edges and assume that $\mathbf{t} \in S_{\text{id}}$ and $t_{|E(\gamma)|}^R \leq t_1$ so that $\mathbf{t} \in A_R^{|E(\gamma)|}$ and that $\mathbf{t} \in (0, 1)^{|E(\gamma)|}$.

There are two cases:  If $\beta$ is in the first case, we can bound $E^{N'}_{\gamma, I, \beta}(\mathbf{t})[\alpha]$ using the arguments leading up to Theorem \ref{thm:f_minus_fNprime}.  More precisely, we have the bound
\begin{align}
  E^{N'}_{\gamma, I, \beta}(\mathbf{t})[\alpha] \preceq \left(\sum_p \prod_{h \in T(\gamma)} \|\alpha\|_{p_h} \right)t_{|E(\gamma)|}^{(N' + 1)/2 + n(|V(\gamma)| - 1)/2 - R(O(\gamma) + n|E(\gamma)|/2)}
\end{align}
However, it remains to bound
\begin{align}
  E^{N',m'}_{\gamma, I, \beta}(\mathbf{t})[\alpha] = |f_{\gamma, I, \beta}^{N'}(\mathbf{t})[\alpha] - f^{N', m'}_{\gamma, I, \beta}(\mathbf{t})[\alpha]|.
\end{align}
By similar reasoning to that used in Proposition \ref{prop:intAKbounds}, we have
\begin{align}
  |\mathcal{I}^{\overline{K}}(\mathbf{t})| \preceq t_{|E(\gamma)|}^{(|\overline{K}| + (n - 1)(|V(\gamma)| - 1))/2} \preceq t_{|E(\gamma)|}^{(n - 1)(|V(\gamma)| - 1)/2} 
\end{align}
Furthermore, 
\begin{align}
  \int_{\RR_{\geq 0}}|\mathcal{I}_{\beta}^{K_n}(w_n, \mathbf{t}) - \mathcal{I}_{\beta}^{K_n}(\infty, \mathbf{t})| w_n^{m + 1} \, dw_n
\end{align}
is less than or equal to 
\begin{align}
  &\int_{\RR_{\geq 0}}\left(\int_{\RR^{|V(\gamma)| - 1} \setminus P(w_n)}e^{-\widetilde{Q}_{\mathbf{t}}^{(\beta)}(y_n)} |y_n^{K_n}|  \, dy_n \right)w_n^{m' + 1}dw_n \\
  &\preceq t_{|E(\gamma)|}^{(|K_n| + m' + 1 + |V(\gamma)|)/2} \leq t_{|E(\gamma)|}^{(m' + 1 + |V(\gamma)|)/2}
\end{align}

Therefore $E^{N', m'}_{\gamma, I, \beta}(\mathbf{t})[\alpha] \preceq$
\begin{align}
   &\sum_{\substack{|\overline{K}| + |K_n| \leq N'\\ -O(\gamma) \leq |J| \leq 0} }\mathbf{t}^{J - n/2}\overline{\mathcal{I}}^{\overline{K}}(\mathbf{t})  \int_{\HH^n} |\mathcal{I}_{\beta}^{K_n}(w_n, \mathbf{t}) - \mathcal{I}_{\beta}^{K_n}(\infty, \mathbf{t})| |c_{J, \overline{K}, K_n}(w) - c^n_{J, \overline{K}, K_n}(w)| \, dw \\
                                                       &\preceq t_{|E(\gamma)|}^{-R(O(\gamma) + n|E(\gamma)|/2)} t_{|E(\gamma)|}^{(n - 1)(|V(\gamma)| - 1)/2} t_{|E(\gamma)|}^{(m' + 1 + |V(\gamma)|)/2}\int_{\RR^{n - 1}} e(\overline{w}) \, d\overline{w}
\end{align}
We have
\begin{align}
  \int e(\overline{w}) \, d\overline{w} \preceq \sum_p \prod_{h \in T(\gamma)} \|\alpha\|_{p_h}
\end{align}
where the summation is over multi-indices $p: T(\gamma) \to \ZZ^{\geq 0}$ such that $\sum_{h \in T(\gamma)} p_h \leq O(\gamma) + N' + m' + 1$.

In conclusion, we have shown that
\begin{thm}\label{thm:f_minus_fNprime_bdry1}
  In the case $\widetilde{c}^{\beta}(e) = 0$ for all $e \in E(\gamma)$,
  \begin{align}
    E^{N', m'}_{\gamma, I, \beta}(\mathbf{t})[\alpha] \preceq \left(\sum_p \prod_{h \in T(\gamma)} \|\alpha\|_{p_h} \right)t_{|E(\gamma)|}^{(m' + 2)/2 + n(|V(\gamma)| - 1)/2 - R(O(\gamma) + n|E(\gamma)|/2)}
  \end{align}
  where the summation is over multi-indices $p: T(\gamma) \to \ZZ^{\geq 0}$ such that $\sum_{h \in T(\gamma)} p_h \leq O(\gamma) + N' + 'm + 1$.
\end{thm}

On the other hand, if $\beta$ is in the second case, we have that $E^{N'}_{\gamma, I, \beta}(\mathbf{t})[\alpha]$ is less than or equal to a constant times $t_{|E(\gamma)|}^{-R(O(\gamma) + n|E(\gamma)|/2)}$ times
\begin{align}
  &\sum_{|\overline{K}| + |K_n| = N' + 1} \int_{\RR^{(n - 1)|V(\gamma)|}} e^{-Q_{\mathbf{t}}(\overline{y})}|\overline{y}^{\overline{K}}| \int_{\HH^n_w}\int_{P(w_n)} e^{ - \widetilde{Q}^{(\beta)}_{\mathbf{t}}(y_n, w_n) - Q^{(\beta)}_{\mathbf{t}}(w_n)}  |y_n^{K_n}| e_{\overline{K}, K_n}(w)\\
  &\preceq t_{|E(\gamma)|}^{(N' + 1)/2 + n(|V(\gamma)| - 1)/2)}\sum_{|K| = N' + 1}\int_{\HH^n_w} e_{\overline{K}, K_n}(\overline{w}, w_n )
\end{align}
Therefore we have the bound
\begin{thm}\label{thm:f_minus_fNprime_bdry2}
  If $\beta$ is in the second case,
  \begin{align}
    E^{N'}_{\gamma, I, \beta}(\mathbf{t})[\alpha] \preceq \left(\sum_p \prod_{h \in T(\gamma)} \|\alpha\|_{p_h} \right)t_{|E(\gamma)|}^{(N' + 1)/2 + n(|V(\gamma)| - 1)/2 - R(O(\gamma) + n|E(\gamma)|/2)}
  \end{align}
  where the summation is over multi-indices $p: T(\gamma) \to \ZZ^{\geq 0}$ such that $\sum_{h \in T(\gamma)} p_h \leq O(\gamma) + N' + 1$.
\end{thm}

Now note that
\begin{align}
  \int_{\RR^{\geq 0}}e^{-Q_{\mathbf{t}}^{(\beta)}(w_n)}\mathcal{I}_{\beta}^{K_n}(w_n, \mathbf{t}) w_n^{m' + 1} \, dw_n
\end{align}
is equal to
\begin{align}
  \int_{\RR_{\geq 0}} e^{-Q_{\mathbf{t}}^{(\beta)}(w_n)}\int_{P(w_n)}e^{-\widetilde{Q}_{\mathbf{t}}^{(\beta)}(y_n, w_n)} |y_n^{K_n}| w_n^{m' + 1} \, dy_ndw_n &\preceq t_{|E(\gamma)|}^{(|K_n| + m' + 1 + V(\gamma))/2} \\
  &\leq t_{|E(\gamma)|}^{(m' + 1 + V(\gamma))/2}
\end{align}
Thus, $E^{N', m'}_{\gamma, I, \beta}(\mathbf{t})[\alpha]$ is less than or equal to a constant times
\begin{align}
  &\sum_{\substack{|\overline{K}| + |K_n| \leq N'\\ -O(\gamma) \leq |J| \leq 0} }t^{J - n/2}\overline{\mathcal{I}}^{\overline{K}}(\mathbf{t})  \int_{\HH^n} e^{-Q_{\mathbf{t}}^{(\beta)}(w_n)}\mathcal{I}_{\beta}^{K_n}(w_n, \mathbf{t}) |c_{J, \overline{K}, K_n}(w) - c^n_{J, \overline{K}, K_n}(w)| \, dw
\end{align}
which is bounded by the same power of $t_{|E(\gamma)|}$ as when $\beta$ is in the second case.
\begin{thm}\label{thm:f_minus_fNprime_bdry3}
  If $\beta$ is in the second case,
  \begin{align}
    E^{N', m'}_{\gamma, I, \beta}(\mathbf{t})[\alpha] \leq \left(\sum_p \prod_{h \in T(\gamma)} \|\alpha\|_{p_h} \right)t_{|E(\gamma)|}^{(m' + 2)/2 + n(|V(\gamma)| - 1)/2 - R(O(\gamma) + n|E(\gamma)|/2)}
  \end{align}
  where the summation is over multi-indices $p: T(\gamma) \to \ZZ^{\geq 0}$ such that $\sum_{h \in T(\gamma)} p_h \leq O(\gamma) + N' + m' + 1$.
\end{thm}

\subsubsection{Inductive Construction of the Counterterms}
\label{sec:induct-constr-count-bdry}
The proof follows along the lines of the proof of Theorem \ref{thm:inductive_step} in Section \ref{sec:induct-constr-count}.  The proof of Theorem \ref{thm:inductive_step}, which was given only in the case $p = 2$ involves two steps.

In the first step, we show that $\left|g_{\gamma', \gamma, I}(\mathbf{t})[\alpha] - g_{\gamma', \gamma, I}^{N_1'}(\mathbf{t})[\alpha]\right|$ is bounded by $t_{|E(\gamma)|}$ to a power that grows linearly in $N_1'$.  This implies that 
\begin{align}
  \left|f_{\gamma, I}(\mathbf{t})[\alpha] - f_{\gamma\setminus\gamma', I, g^{N_1'}_{\gamma', \gamma, I}}(\mathbf{t})[\alpha]\right|
\end{align}
is bounded by a power of $t_{|E(\gamma)|}$ that grows linearly in $N_1'$.
Now we also have $\left|g_{\gamma', \gamma, I}^{N_1'}(\mathbf{t})[\alpha] - g_{\gamma', \gamma, I}^{N_1', m_1'}(\mathbf{t})[\alpha]\right|$ is bounded by a power of $t_{|E(\gamma)|}$ that grows linearly in $m_1'$.  This implies that 
\begin{align}
  \left|f_{\gamma\setminus\gamma', I, g^{N_1'}_{\gamma', \gamma, I}}(\mathbf{t})[\alpha] - f_{\gamma\setminus\gamma', I, g^{N_1', m_1'}_{\gamma', \gamma, I}}(\mathbf{t})[\alpha]\right|
\end{align}
is bounded by a power of $t_{|E(\gamma)|}$ that grows linearly in $m_1'$.

In the second step, we show that 
\begin{align}
  \left|f_{\gamma \setminus \gamma', I, g^{N_1', m_1'}_{\gamma', \gamma, I}}(\mathbf{t})[\alpha] - f^{N_2'}_{\gamma \setminus \gamma', I, g^{N_1', m_1'}_{\gamma', \gamma, I}}(\mathbf{t})[\alpha]\right|
\end{align}
is bounded by a power of $t_{|E(\gamma)|}$ that grows linearly in $N_2'$ for $N_1'$ and $m_1'$ fixed.
Now we also have to show that
\begin{align}
  \left|f^{N_2'}_{\gamma\setminus\gamma', I, g^{N_1', m_1'}_{\gamma', \gamma, I}}(\mathbf{t})[\alpha] - f^{N_2', m_2}_{\gamma\setminus\gamma', I, g^{N_1', m_1'}_{\gamma', \gamma, I}}(\mathbf{t})[\alpha]\right|.
\end{align}
is bounded by a power of $t_{|E(\gamma)|}$ that grows linearly in $m_2$ for $N_2'$, $N_1'$ and $m_1'$ fixed.

In all bounds except the first, a potential problem is we now have to contend with the fact that $g^{N_1', m_1'}_{\gamma', \gamma, I}$ involves local functionals on the boundary $\RR^{n - 1}$.  But a local functional on $\RR^{n - 1}$ can be rewritten as a local functional on $\HH^n$ with total derivative integrand.

\subsection{Global Counterterms on a Curved Manifold}\label{sec:count-comp-manif}

\subsubsection{Replacing the Heat Kernel with Its Asymptotic Series and Bounding the Error}
\label{sec:rep_heat_asymp}

On a compact manifold without boundary, we have the asymptotic series for the scalar heat kernel
\begin{align}
  K_t(x, y) \sim (4\pi t)^{-n/2}e^{-d(x, y)^2/4t}\sum_{i = 0}^{\infty} \phi_i(x, y)t^i.
\end{align}
Let
\begin{align}
  K_t^{[N]}(x, y) = (4\pi t)^{-n/2}e^{-d(x, y)^2/4t}\sum^N_{i = 0} \phi_i(x, y)t^i.
\end{align}
and recall that the smooth functions $\phi_i(x, y)$ can be taken to be supported in an arbitrarily small neighborhood $U$ of the diagonal.
Beginning from \eqref{eq:f_gamma_cpct}, the formula for $f_{\gamma, I}(\mathbf{t})[\alpha]$, for each edge, we replace $K_t$ with $K_t^{[N]}$.  We let $f_{\gamma, I}^{[N]}$ denote the result of making all $|E(\gamma)|$ of such substitutions.  

An analogous statement to Proposition \ref{thm:derivatives-k_t} for the flat case can be made for $K_t^{[N]}$:
\begin{prop}
 For multi-indices $I = (i_1, \dots, i_n)$ and $I' = (i'_1, \dots, i'_n)$, $\frac{\partial K_t^{[N]}}{\partial x^I \partial x^{I'}}$ is equal to a Laurent polynomial in $t$ with coefficients in smooth functions of $x$ and $y$
 \begin{align}
   P_{I, I'}(x, y, 1/t) = \sum_{k = -|I| - |I'|}^N f_k(x, y)t^k
 \end{align}
multiplied by $(4\pi t)^{-n/2} e^{-d(x, y)^2/4t}$.
\end{prop}

Therefore, we have a formula 
\begin{align}
  f_{\gamma, I}^{[N]}(\mathbf{t})[\alpha] = \sum_{-O(\gamma) \leq |J| \leq |E(\gamma)|N }  \mathbf{t}^{J - n/2}\int_{D_{U, \gamma}} e^{-\sum_{e \in E(\gamma)} Q^{\text{nl}}_e/4 t_e} \Phi_J
\end{align}
where $Q^{\text{nl}}_e = d^2(x_{v_1(e)}, x_{v_2(e)})$, the square of geodesic distance.

The property of the asymptotic series for the scalar heat kernel implies in particular that for $t \in (0, 1)$
\begin{align}
  \|K_t - K_t^{[N]}\|_l \preceq t^{N - n/2 -l},
\end{align}
where $\|\cdot\|_l$ is the $C^l$ norm on $M \times M$.
Clearly, for $t \in (0, 1)$, $\|K_t^{[N]}\|_l \preceq t^{-n/2 - l}$.  Also note that an application of the triangle inequality using the asymptotic formula for $N = 0$ shows that $\|K_t\|_l \preceq t^{-n/2 - l}$.  
Assume that $\mathbf{t} \in S_{id} \cap (0, 1)^{|E(\gamma)|}$; recall that this means that we've ordered the edges and assumed that $t_1 \leq \dots \leq t_{|E(\gamma)|}$.  Then
\begin{align}
  |f_{\gamma, I}(\mathbf{t})[\alpha] - f_{\gamma, I}^{[N]}(\mathbf{t})[\alpha]| \preceq \sum_p\|\alpha\|^{|T(\gamma)|}_{O(\gamma)}\sum_{i = 1}^{|E(\gamma)|} t_i^{N - n/2 - p_i}\prod_{j \neq i} t_j^{- n/2 - p_j} 
\end{align}
where the sum is over collections $p$ of $|E(\gamma)|$ nonnegative integers $p_j$ with $\sum_j p_j \leq O(\gamma)$.

Now suppose that $t_{|E(\gamma)|}^R \leq t_1$.  Then
\begin{align}
  |f_{\gamma, I}(\mathbf{t})[\alpha] - f_{\gamma, I}^{[N]}(\mathbf{t})[\alpha]| \preceq \|\alpha\|^{|T(\gamma)|}_{O(\gamma)} t_{|E(\gamma)|}^{N - R(O(\gamma) + n|E(\gamma)|/2)}.
\end{align}

On a noncompact manifold, where the heat kernel is not unique, our starting point is what might be called a ``fake heat kernel'' $K_t^{\text{fake}}(x, y)$ and the propagator $P_{\epsilon}^L = \int_{\epsilon}^LK_t^{\text{fake}}(x, y)$.  By construction, the fake heat kernel is of the form $K_t^{\text{fake}}(x, y) = K_t^{[N]}(x, y)$ for some large $N$ and satisfies the equation $(\partial_t + D_x)K_t^{\text{fake}}(x, y) = C(t, x, y)$, where $D_x$ is the Laplacian acting in the $x$ direction and where $C(t, x, y)$ is a function on $\RR_{\geq 0} \times M \times M$ with $\lim_{t \to 0^+} C(t, x, y) = 0$.

On any smooth manifold without boundary, because $K^{[N]}_t(x, y)$ is supported in $U$, have a formula 
\begin{align}
  f_{\gamma, I}^{[N]}(\mathbf{t})[\alpha] = \sum_{-O(\gamma) \leq |J| \leq |E(\gamma)|N }  \mathbf{t}^{J - n/2}\int_{D_{U, \gamma}} e^{-\sum_{e \in E(\gamma)} Q^{\text{nl}}_e/4 t_e} \Phi_J
\end{align}
where $Q^{\text{nl}}_e = d^2(x_{v_1(e)}, x_{v_2(e)})$, the square of geodesic distance and where
\begin{align}
  D_{U, \gamma} := \{(x_1, \dots, x_{|V(\gamma)|}) \in M^{|V(\gamma)|} : \text{for all } e \in E(\gamma), (x_{v_1(e)}, x_{v_2(e)}) \in U \}.
\end{align}
We would like to cover $D_{U, \gamma}$ by coordinate neighborhoods.

\subsubsection{Taylor Expanding $\Phi_J$ and Bounding the Error}
While it is possible to construct a geometric version of spanning tree coordinates, we shall do something simpler.  For each point $x_0 \in M$, let $B_{x_0}$ be a geodesic ball around $x_0$.  Then the open sets $B_{x_0}^{|V(\gamma)|}$ form a cover of $D_{U, \gamma}$ by coordinate neighborhoods, for $U$ sufficiently small.  We denote these coordinates $\{x_v\}_{v \in V(\gamma)}$.  One can then make a linear change of coordinates to spanning tree coordinates (for example, of the distinguished vertex flavor).

Taking the Taylor expansion of $\Phi$ with respect to $\{y_e\}_{e \in E(T)}$, $f^{[N], N'}_{\gamma, I}(\mathbf{t})[\alpha]$ is equal to 
\begin{align}
  &\sum_{\substack{|K| \leq N'\\ -O(\gamma) \leq |J| \leq N|E(\gamma)|} } \mathbf{t}^{J - n/2}\int_{D_{U_{\delta}, \gamma}} e^{-\sum_{e \in E(\gamma)} Q^{\text{nl}}_e(w, y)/4 t_e} c_{J, K}(w) y^K\, dy dw\\
  = &\sum_{\substack{|K| \leq N', \text{$K$ even}\\ -O(\gamma) \leq |J| \leq N|E(\gamma)|} }\mathbf{t}^{J - n/2}  \int_B c_{J, K}(w) \mathcal{I}^K(w, \mathbf{t})\, dw
\end{align}
This differs from the case of $\RR^n$ in that
\begin{align}
  \mathcal{I}^K(w, \mathbf{t}) = \int_{\RR^{n(|V(\gamma)| - 1)}} e^{-Q^{\text{nl}}_{\mathbf{t}}(w, y)} y^K \, dy.
\end{align}
depends $w$ because $Q^{\text{nl}}_{\mathbf{t}}(w, y)$ depends on $w$ and thus $\int_U  c_{J, K}(w) \mathcal{I}^K(w, \mathbf{t})\, dw$ is not always a local functional multiplied by a function of $\mathbf{t}$.  However, in the case that $M$ is a Riemannian locally symmetric space, $Q^{\text{nl}}_{\mathbf{t}}(w, y) = Q^{\text{nl}}_{\mathbf{t}}(y) $ is independent of $w$, $\mathcal{I}_A^K(w, \mathbf{t}) = \mathcal{I}_A^K(\mathbf{t})$ comes out of the integral, and we produce local functionals $\int_U  c_{J, K}(w)$, as desired.  Note that $Q^{\text{nl}}_{\mathbf{t}}(y)$ is not necessarily a polynomial in the $y_e$ variables, so there might not be a way to compute $\mathcal{I}_A^K(\mathbf{t})$ in general on a Riemannian locally symmetric space.

For a general Riemannian manifold without boundary, to produce local functionals, we can replace $Q^{\text{nl}}_{\mathbf{t}} = \sum_{e \in E(\gamma)} Q^{\text{nl}}_e/4t_e$ with its Taylor polynomial $Q_{\mathbf{t}}^m$ of degree $m$.  Recall that we have previously defined $Q_e = |x_{v_1(e)} - x_{v_2(e)}|^2$.  The Taylor polynomial splits into its degree $2$ and higher degree parts
\begin{align}
  Q_{\mathbf{t}}^m(x) = Q_{\mathbf{t}}(x) + \widetilde{Q}_{\mathbf{t}}^m(x)
\end{align}
with
\begin{align}
  Q_{\mathbf{t}}(x) = \sum_{e \in E(\gamma)} Q_e(x)/4t_e
\end{align}
and
\begin{align}
  \widetilde{Q}_{\mathbf{t}}^m(x) = \sum_{e \in E(\gamma)} \widetilde{Q}^m_e(x)/4t_e = O\left(\left(\sum_{v \in V(\gamma)} |x_v|\right)^3\right).
\end{align}
We define $f_{\gamma, I}^{[N, m]}(\mathbf{t})[\alpha]$ to be the result of substituting $Q_{\mathbf{t}}^m(x)$ for $Q^{\text{nl}}_{\mathbf{t}}(x)$ in $f_{\gamma, I}^{[N]}(\mathbf{t})[\alpha]$.

Suppose that $A \geq 0$ and $B = B_0 + B_1$ with $B_0 \geq 0$.  From the mean value theorem, we have the inequality
\begin{align}
  |e^{-A} - e^{-B}| \leq e^{\max(-A, -B)}|A - B| \leq \max(1, e^{-B_1})|A - B| \leq e^{|B_1|}|A - B|.
\end{align}
With $A = Q^{\text{nl}}_{\mathbf{t}}(x)$, $B_0 = Q_{\mathbf{t}}(x)$ and $B_1 = \widetilde{Q}_{\mathbf{t}}^m(x)$, we have the inequality $|f_{\gamma, I}^{[N]}(\mathbf{t})[\alpha] - f_{\gamma, I}^{[N, m]}(\mathbf{t})[\alpha]| \leq$
\begin{align}
   &\sum_{-O(\gamma) \leq |J| \leq |E(\gamma)|N }  \mathbf{t}^{J - n/2}\int_{\RR^{n|V(\gamma)|}} \left| e^{-Q^{\text{nl}}_{\mathbf{t}}(x)} - e^{-Q^m_{\mathbf{t}}(x)} \right|\Phi_J(x)\\
   \leq &\sum_{-O(\gamma) \leq |J| \leq |E(\gamma)|N }  \mathbf{t}^{J - n/2}\int_{\RR^{n|V(\gamma)|}} e^{|\widetilde{Q}^m_{\mathbf{t}}(x)|}|Q^{\text{nl}}_{\mathbf{t}}(x) - Q^m_{\mathbf{t}}(x)|\Phi_J(x)
\end{align}
For $t_{|E(\gamma)|}^R \leq t_1$, we change of coordinates $x_v$ goes to $\sqrt{t_{|E(\gamma)|}}x_v$, and we have $|\widetilde{Q}^m_{\mathbf{t}}(\sqrt{t_{|E(\gamma)|}}x)| \leq C\sqrt{t_{|E(\gamma)|}}$ on $\overline{B}^{|V(\gamma)|}$ and thus $|\widetilde{Q}^m_{\mathbf{t}}(\sqrt{t_{|E(\gamma)|}}x)| \leq C$, since we are assuming that $t_{|E(\gamma)|} \in (0, 1)$.
Finally, we have the bound for $E^{[N, m]}(\mathbf{t})[\alpha] := |f_{\gamma, I}^{[N]}(\mathbf{t})[\alpha] - f_{\gamma, I}^{[N, m]}(\mathbf{t})[\alpha]|$:
\begin{prop}
\begin{align}
  E^{[N, m]}(\mathbf{t})[\alpha] \leq\|\alpha\|_{O(\gamma)}^{|T(\gamma)|}t_{|E(\gamma)|}^{(m + 1)/2 - R + n|V(\gamma)|/2 - R(O(\gamma) + n|E(\gamma)|/2)} 
\end{align}
  
\end{prop}

We then make the linear change of coordinates to spanning tree coordinates $w$, $\{y_e\}$ and form $f^{[N, m], N'}_{\gamma, I}(\mathbf{t})[\alpha]$ by replacing $\Phi_J$ in $f^{[N, m]}_{\gamma, I}(\mathbf{t})[\alpha]$ with its Taylor expansion in $\{y_e\}$ to order $N'$.  The bound on 
\begin{align}
  |f^{[N, m]}_{\gamma, I}(\mathbf{t})[\alpha] - f^{[N, m], N'}_{\gamma, I}(\mathbf{t})[\alpha]|
\end{align}
can be established as in the proof of Theorem \ref{thm:f_minus_fNprime}.

\subsubsection{Inductive Construction of the Counterterms}
\label{sec:induct-constr-count-curved}
The construction works similarly to the one given in Section \ref{sec:induct-constr-count}.  The proof of Theorem \ref{thm:inductive_step}, which was given only in the case $p = 2$ involves two steps, which we modify as follows.

In the first step, we begin by showing that $\left|g_{\gamma', \gamma, I}(\mathbf{t})[\alpha] - g_{\gamma', \gamma, I}^{[N_1]}(\mathbf{t})[\alpha]\right|$ is bounded by $t_{|E(\gamma)|}$ to a power that grows linearly in $N_1$.  This implies that 
\begin{align}
  \left|f_{\gamma, I}(\mathbf{t})[\alpha] - f_{\gamma\setminus\gamma', I, g^{[N_1]}_{\gamma', \gamma, I}}(\mathbf{t})[\alpha]\right|
\end{align}
is bounded by a power of $t_{|E(\gamma)|}$ that grows linearly in $N_1$.
We then also show that $\left|g_{\gamma', \gamma, I}^{[N_1]}(\mathbf{t})[\alpha] - g_{\gamma', \gamma, I}^{[N_1, m_1]}(\mathbf{t})[\alpha]\right|$ is bounded by $t_{|E(\gamma)|}$ to a power that grows linearly in $m_1$ for $N_1$ fixed.  This implies that 
\begin{align}
  \left|f_{\gamma\setminus\gamma', I, g^{[N_1]}_{\gamma', \gamma, I}}(\mathbf{t})[\alpha] - f_{\gamma\setminus\gamma', I, g^{[N_1, m_1]}_{\gamma', \gamma, I}}(\mathbf{t})[\alpha]\right|
\end{align}
is bounded by a power of $t_{|E(\gamma)|}$ that grows linearly in $m_1$ for $N_1$ fixed.
The rest of the step works like the local situation in Section \ref{sec:induct-constr-count}.

In the second step, we begin by showing that 
\begin{align}
  \left|f_{\gamma \setminus \gamma', I, g^{[N_1, m_1], N_1', m_1'}_{\gamma', \gamma, I}}(\mathbf{t})[\alpha] - f^{[N_2]}_{\gamma \setminus \gamma', I, g^{[N_1, m_1], N_1', m_1'}_{\gamma', \gamma, I}}(\mathbf{t})[\alpha]\right|
\end{align}
is bounded by a power of $t_{|E(\gamma)|}$ that grows linearly in $N_2$ for $N_1$, $m_1$, $N_1'$ and $m_1'$ fixed.
We then also show that
\begin{align}
  \left|f^{[N_2]}_{\gamma\setminus\gamma', I, g^{[N_1, m_1], N_1', m_1'}_{\gamma', \gamma, I}}(\mathbf{t})[\alpha] - f^{[N_2, m_2]}_{\gamma\setminus\gamma', I, g^{[N_1, m_1], N_1', m_1'}_{\gamma', \gamma, I}}(\mathbf{t})[\alpha]\right|.
\end{align}
is bounded by a power of $t_{|E(\gamma)|}$ that grows linearly in $m_2$ for $N_1$, $m_1$, $N_1'$, $m_1'$ and $N_2$ fixed.
The rest of the step works like the local situation in Section \ref{sec:induct-constr-count}.

\subsection{Global Counterterms on a Curved Manifold with Boundary}\label{sec:count-comp-manif-bdry}
In this section, we argue that the renormalization procedure can also be carried out in the case of a Riemannian manifold with boundary $M$ with a cylindrical collar; that is, we assume there exists a neighborhood $W$ of $\partial M$ that is isometric to a product $\partial M \times [0, \epsilon)$.  We conjecture that the result is also true for all Riemannian manifolds with boundary $M$ that are doubleable (see Appendix \ref{sec:riemdouble}).

When we have such a manifold with boundary $M$ equipped with a cylindrical collar, the double of $M$, which will be denoted by $M'$, will be a smooth manifold without boundary equipped with an involution $p \mapsto p^*$ that sends a point $p$ to its reflection through the boundary.

The Dirichlet heat kernel on $M$ is
\begin{align}
  K_t(x, y) = K'_t(x, y) - K'_t(x, y^*),
\end{align}
and the Neumann heat kernel on $M$ is 
\begin{align}
  K_t(x, y) = K'_t(x, y) + K'_t(x, y^*),
\end{align}
where $K'_t(x, y)$ is the heat kernel on $M'$.
The existence of an asymptotic expansion of $K'_t(x, y)$ implies that
\begin{align}
  K_t(x, y)\sim e^{-d(x, y)^2/4t}\sum_i \phi_i(x, y)t^i + e^{-d(x, y^*)^2/4t}\sum_i \psi_i(x, y) t^i
\end{align}
where $\psi_i(x, y) = - \phi_i(x, y^*)$ in the Dirichlet case and $\psi_i(x, y) = \phi_i(x, y^*)$ in the Neumann case.
 
Since the cylindrical neighborhood $W$ is a product, by the Pythagorean theorem the square distance is
\begin{align}
  d_M^2(x, y) = d_{\partial M}^2(\overline{x}, \overline{y}) + |x_n - y_n|^2
\end{align}
Therefore, on open sets of the form $B \times [0, \epsilon)$, for $B$ a geodesic ball in $\partial M$, we can use the results of \ref{sec:count_upp_half} for the direction normal to the boundary and use the analysis of \ref{sec:count-comp-manif} for the $\partial M$ direction.

On the complement $M \setminus W$, we have
\begin{align}
  K_t(x, y)\sim e^{-d(x, y)^2/4t}\sum_i \phi_i(x, y)t^i
\end{align}
so we can apply the analysis of \ref{sec:count-comp-manif}.

\section{Construction of an Effective Field Theory from a Local Functional}
\label{sec:constr_eft}

\subsection{Renormalization of Feynman Weights}
\label{sec:reg_feyn_weights}
In this section, following \cite{2007arXiv0706.1533C}, we lay the groundwork for Costello's procedure that constructs an effective action from a local functional $I \in \mathcal{O}(\mathcal{E})[[\hbar]]$, using as input the constructions of Sections \ref{sec:induct-constr-count}, \ref{sec:induct-constr-count-bdry}, and \ref{sec:induct-constr-count-curved}.

On $S_{id}$, we do the following for each sequence $I^{(1)}, \dots, I^{(p)}$ as in Corollary \ref{cor:refined_cover}: Let $\mathcal{N}_1$, \dots, $\mathcal{N}_p$ be a choice of integers so $d_1(\mathcal{N}_1) + d_2(\mathcal{N}_1, \mathcal{N}_2) + \dots + d_p(\mathcal{N}_1, \dots, \mathcal{N}_p) \geq d$ for some fixed positive integer $d$.  Then, for $\mathbf{t} \in S_{id} \cap (0, 1)$,
  \begin{align}
    |f_{\gamma, I}(\mathbf{t})[\alpha] - f^{\mathcal{N}_1, \dots, \mathcal{N}_p}_{\gamma, I}(\mathbf{t})[\alpha]| \preceq \|\alpha\|_q^{|T(\gamma)|}t_{|E(\gamma)|}^d.
  \end{align}
Let
\begin{align}
  \overline E^{I^{(1)}, \dots, I^{(p)}}_R(t_{|E_{\gamma}|}) = \overline{E}^{I^{(1)}}_{R} \cap \overline{E}^{I^{(2)}}_{R}\dots \cap \overline{E}^{I^{(p)}}_{R} \cap (\epsilon, L)^{|E(\gamma)|}
\end{align}
where we think of the first $|E(\gamma)| - 1$ variables to be varying and the last variable fixed to the value $t_{|E(\gamma)|}$.
Consider the quantities
\begin{align}
  F_{\gamma, I}(t_{|E(\gamma)|})[\alpha] &= \int_{\overline E^{I^{(1)}, \dots, I^{(p)}}_R(t_{|E(\gamma)|})}f_{\gamma, I}(\mathbf{t})[\alpha] \, dt_1 \dots dt_{|E(\gamma)| - 1}\\
  F_{\gamma, I}^{\mathcal{N}_1, \dots, \mathcal{N}_p}(t_{|E(\gamma)|})[\alpha] &= \int_{\overline E^{I^{(1)}, \dots, I^{(p)}}_R(t_{|E(\gamma)|})}f^{\mathcal{N}_1, \dots, \mathcal{N}_p}_{\gamma, I}(\mathbf{t})[\alpha] \, dt_1 \dots dt_{|E(\gamma)| - 1}
\end{align}
Then
\begin{align}
  &\left|\int_{\epsilon}^LF_{\gamma, I}[\alpha] - \int_{\epsilon}^LF_{\gamma, I}^{\mathcal{N}_1, \dots, \mathcal{N}_p}[\alpha] - \int_0^L\left(F_{\gamma, I}[\alpha] - F_{\gamma, I}^{\mathcal{N}_1, \dots, \mathcal{N}_p}[\alpha]\right)\right| \\
  \preceq &\|\alpha\|_q^{|T(\gamma)|}\int_0^{\epsilon}\left(\int_{\overline E^{I^{(1)}, \dots, I^{(p)}}_R(t_{|E(\gamma)|})}\right)t_{|E(\gamma)|}^d
\end{align}
We can make the same construction on $S_{\sigma}$ for any permutation $\sigma$.

Let $w^{\text{CT}}_\gamma(P^L_\epsilon, I)[\alpha]$ be the sum over all permutations $\sigma$ and all sequences $I^{(1)}, \dots, I^{(p)}$ of the quantity
\begin{align}
   \int_{\epsilon}^LF_{\gamma, I}^{\mathcal{N}_1, \dots, \mathcal{N}_p}[\alpha] - \int_0^L\left(F_{\gamma, I}[\alpha] - F_{\gamma, I}^{\mathcal{N}_1, \dots, \mathcal{N}_p}[\alpha]\right).
\end{align}
Write
\begin{align}
  w^{\text{CT}}_\gamma(P^L_\epsilon, I)[\alpha] = \sum_i h_i(\epsilon)\Psi_i[\alpha] + \sum_j \Phi_j(L)[\alpha]
\end{align}
with $\sum_i h_i(\epsilon)\Psi_i[\alpha]$ coming from terms of the form $\int_{\epsilon}^1F_{\gamma, I}^{\mathcal{N}_1, \dots, \mathcal{N}_p}[\alpha]$.

Then, as a consequence of the structure of $\int_1^LF_{\gamma, I}^{\mathcal{N}_1, \dots, \mathcal{N}_p}[\alpha]$, we have that $\Phi_i(L)[\alpha]$ is asymptotically local, for small $L$.  Importantly, we also have that for $\epsilon \in (0, 1)$,
\begin{align}\label{eq:wgammact}
  |w_\gamma(P_\epsilon^L, I)[\alpha] -  w^{\text{CT}}_\gamma(P_\epsilon^L, I)[\alpha] | &\preceq \|\alpha\|^{|T(\gamma)|}_q\epsilon^{|E(\gamma)| + d}.
\end{align}

Let $\mathcal{F}$ be the space of functions of $\epsilon$ that is spanned by the $h_i(\epsilon)$ and constant functions.  Let $\mathcal{F}_0$ be the subspace of $\mathcal{F}$ having a limit when $\epsilon \to 0^+$ and choose a complementary subspace of $\mathcal{F}_0$, which we call $\mathcal{F}_-$.  We write $\Reg: \mathcal{F} \to \mathcal{F}_0$ and $\Sing: \mathcal{F} \to \mathcal{F}_-$ for the projection maps. 

An important point is that
\begin{align}
  \lim_{\epsilon \to 0^+}[w_\gamma(P^L_\epsilon, I) - \Sing w^{\text{CT}}_\gamma(P^L_\epsilon, I)]
\end{align}
exists because of (\ref{eq:wgammact}) and the existence of
\begin{align}
  \lim_{\epsilon \to 0^+}\Reg w_\gamma(P^L_\epsilon, I).
\end{align}

\subsection{Costello's Construction of the Effective Interaction}
\label{sec:cost_construction}
Here we more closely follow \cite{Costello_2011}, where Costello begins with the base case of the induction by constructing $I_{1, 1}^{\text{CT}}(\epsilon) := \Sing W_{1, 1}(P_{\epsilon}^L, I)$,
where
\begin{align}
W^{\text{CT}}_{1, 1}(P_\epsilon^L, I) = \sum_{\substack{\gamma\text{ conn}\\ g(\gamma) = 1, T(\gamma) = 1}}\frac{1}{|\Aut(\gamma)|}\hbar^{g(\gamma)} w^{\text{CT}}_\gamma\left(P_\epsilon^L, I\right).
\end{align}
It is clear from the analysis of the previous section that $I_{1, 1}[L] = \lim_{\epsilon \to 0^+}[W_{1, 1}(P_{\epsilon}^L, I) - I_{1, 1}^{\text{CT}}(\epsilon)]$ is well-defined, and that $I_{1, 1}[L]$ is asymptotically local for small $L$.

Suppose that $I_{i', k'}(\epsilon)$ has been constructed for all $(i', k') \prec (i, k)$, so that
\begin{align}
  I_{\prec (i, k)}[L] = \lim_{\epsilon \to 0^+} W_{\prec (i, k)}\left(P_{\epsilon}^L, I - \sum_{(i', k') \prec (i, k)}\hbar^{i'}I_{i', k'}(\epsilon)\right).
\end{align}
Costello then defines
\begin{align}
  I^{\text{CT}}_{i, k}(\epsilon, L) = \Sing W_{i, k}^{CT}\left(P_\epsilon^L,  I - \sum_{(i', k') \prec (i, k)}\hbar^{i'}I_{i', k'}(\epsilon)\right), 
\end{align}
and uses the identity
\begin{align}
  W_{i, k}\left(P_\epsilon^L,  I - \sum_{(i', k') \prec (i, k)}\hbar^{i'}I_{i', k'}(\epsilon) - \hbar^i I^{\text{CT}}_{i, k}(\epsilon, L)\right) =& \\
  W_{i, k}\left(P_\epsilon^L,  I - \sum_{(i', k') \prec (i, k)}\hbar^{i'}I_{i', k'}(\epsilon) \right)& - I^{\text{CT}}_{i, k}(\epsilon, L)
\end{align}
to show that the limit
\begin{align}
  I_{i, k}[L] = \lim_{\epsilon \to 0^+}W_{i, k}\left(P_\epsilon^L,  I - \sum_{(i', k') \prec (i, k)}\hbar^{i'}I_{i', k'}(\epsilon) - \hbar^iI^{\text{CT}}_{i, k}(\epsilon, L)\right)
\end{align}
exists.

Costello then shows that $I^{\text{CT}}_{i, k}(\epsilon, L)$ does not depend on $L$ using the identity
\begin{align}
  I^{\text{CT}}_{i, k}(\epsilon, L) &= \Sing W_{i, k}^{CT}\left(P_{L'}^{L},  W_{\prec (i, k)}\left(P_{\epsilon}^{L'}, I - \sum_{(i', k') \prec (i, k)}\hbar^{i'}I_{i', k'}(\epsilon)\right)\right) \\
                                    &+ \Sing W_{i, k}^{CT}\left(P_{\epsilon}^{L'}, I - \sum_{(i', k') \prec (i, k)}\hbar^{i'}I_{i', k'}(\epsilon)\right)\\
  &= I^{\text{CT}}_{i, k}(\epsilon, L')
\end{align}

\subsection{Statement of the Main Theorem}
\label{sec:statement-main-thm}
\begin{defin}
  We say an effective scalar field theory on a Riemannian manifold without boundary $M$ is of type:
  \begin{itemize}
  \item \emph{E} if $M = \RR^n$, the space of fields is $\mathcal{S}(\RR^n)$, the space of Schwartz functions, and the propagator is $P_{\epsilon}^L = \int_{\epsilon}^LK_t$, where $K_t$ is the Euclidean heat kernel.  All functionals may be assumed to be translation invariant.
  \item \emph{C} if $M$ is compact, the space of fields is $C^{\infty}(M)$, and the propagator is $P_{\epsilon}^L = \int_{\epsilon}^LK_t$ where $K_t$ is the heat kernel on $M$.
  \item \emph{NC} if the space of fields is $C^{\infty}_c(M)$, and the propagator is $P_{\epsilon}^L = \int_{\epsilon}^LK_t^{\text{fake}}$ where $K_t^{\text{fake}}$ is a fake heat kernel on $M$.  The effective action is well-defined modulo constant functionals.
  \end{itemize}
\end{defin}
\begin{defin}
  We say an effective scalar field theory on a Riemannian manifold with boundary (with cylindrical collar) $M$ is of type:
  \begin{itemize}
  \item \emph{E} if $M = \HH^n$, the space of fields is $\mathcal{S}_D(\HH^n)$ (or $\mathcal{S}_N(\HH^n)$), the space of Schwartz functions on $\HH^n$ with Dirichlet (or Neumann) boundary conditions, and the propagator is $P_{\epsilon}^L = \int_{\epsilon}^LK_t$, where $K_t$ is the Dirichlet (or Neumann) heat kernel on $M$.  All functionals may be assumed to be translation invariant parallel to the boundary.
  \item \emph{C} if $M$ is compact and the space of fields is $C^{\infty}_D(M)$ (or $C^{\infty}_N(M)$), and the propagator is $P_{\epsilon}^L = \int_{\epsilon}^LK_t$, where $K_t$ is the Dirichlet (or Neumann) heat kernel on $M$.
  \item \emph{NC} if the space of fields is $C^{\infty}_{D, c}(M)$ (or $C^{\infty}_{N, c}(M)$), and the propagator is $P_{\epsilon}^L = \int_{\epsilon}^LK_t^{\text{fake}}$, where $K_t^{\text{fake}}$ is a fake Dirichlet (or Neumann) heat kernel on $M$.  The effective action is well-defined modulo constant functionals.
  \item \emph{C-E} if $M = N_0 \times \RR_{\geq 0}$ and $N_0$ is compact, the space of fields is $C^{\infty}(N_0) \widehat{\otimes} \mathcal{S}_D(\RR_{\geq 0})$ (or $C^{\infty}(N_0) \widehat{\otimes} \mathcal{S}_N(\RR_{\geq 0})$), and the propagator is $P_{\epsilon}^L = \int_{\epsilon}^LK_t$, where $K_t = K_t^{N_0}K_t^{\RR_{\geq 0}}$, for $K_t^{\RR_{\geq 0}}$ the Dirichlet (or Neumann) heat kernel on $\RR_{\geq 0}$.
  \item \emph{NC-E} if $M = N_0 \times \RR_{\geq 0}$, the space of fields is $C^{\infty}_c(N_0) \widehat{\otimes} \mathcal{S}_D(\RR_{\geq 0})$ (or $C^{\infty}_c(N_0) \widehat{\otimes} \mathcal{S}_N(\RR_{\geq 0})$), and the propagator is $P_{\epsilon}^L = \int_{\epsilon}^LK_t$, where $K_t = K_t^{N_0, \text{fake}}K_t^{\RR_{\geq 0}}$, for $K_t^{\RR_{\geq 0}}$ the Dirichlet (or Neumann) heat kernel on $\RR_{\geq 0}$.  The effective action is well-defined modulo constant functionals.
  \end{itemize}
\end{defin}
Note that the reason for disallowing constant functionals in the effective action for the NC types is to ensure that the Feynman weights are well-defined despite the (fake) heat kernel not being compactly-supported. 

We can now state our slight modification and generalization of the statement of Costello \cite{Costello_2011}
\begin{thm}
  For any of the above types of scalar field theories, given the space of fields $\mathcal{E}$ and a local interaction functional $I \in \mathcal{O}_{\text{loc}}(\mathcal{E})[[\hbar]]$, there exists a functional of counterterms $I^{\text{CT}}(\epsilon)$ such that the effective interaction $I[L] = \lim_{\epsilon \to 0^+}[W(P_\epsilon^L, I - I^{\text{CT}}(\epsilon))]$ exists and satisfies
  \begin{itemize}
  \item The renormalization group equation: $I[L] = W(P_{\epsilon}^L, I[\epsilon])$,
  \item Asymptotic locality: for each $d$ and $(i, k)$, there exists a functional
    \begin{align}
      I_{i, k}^{\text{loc}, d}[L] = \sum_j h^j_{i, k}(L) \Psi_j
    \end{align}
    so that
    \begin{align}
      \left|I_{i, k}[L][\alpha] - I_{i, k}^{\text{loc}, d}[L][\alpha] \right| \preceq \|\alpha\|^k_q L^{d + E(i, k)}
    \end{align}
    for some positive integers $q$ and $E(i, k)$.
  \end{itemize}

\end{thm}

\begin{appendices}

\section{The Double of a Riemannian Manifold}
\label{sec:riemdouble}
The double $M'$ of an $n$-manifold with boundary $M$ is defined to be the disjoint union of two copies of $M$ (which we index by $0$ and $1$), with the identification of $(p, 0) \sim (p, 1)$ for $p \in \partial M$.  If $M$ is a smooth manifold with collar neighborhood $U \cong \partial M \times [0, \epsilon)$, there is a smooth structure on $M'$ such that $U' \cong M \times (-\epsilon, \epsilon)$ is a diffeomorphism, and we call $M'$ the smooth double.

If $(M, g)$ is a Riemannian manifold, there is a canonical smooth collar arising as the domain $U$ of the normal exponential map $\exp^{\perp}$ and we construct the smooth double $M'$ using this collar.  However, it is not always true that $g'$, the natural continuous extension of $g$ to $M'$, is smooth.  We recall the proof of a result of H. Mori \cite{MR1145141} that characterizes exactly when $g'$ is smooth.

Before proceeding we set conventions and recall a few basic facts.  Let $N$ denote the inward pointing unit normal vector field along $\partial M$.  For any coordinate system $x^1, \dots, x^{n - 1}$ on a neighborhood $V \subset \partial M$, there is a canonical coordinate system $x^1, \dots, x^n$ on $V \times [0, \epsilon)$ with the property $x^n(\exp^{\perp}(tN)) = t$ for $t \in [0, \epsilon)$.  This extends to a coordinate system on $V \times (-\epsilon, \epsilon)$ with $x^n((\exp^{\perp}(tN), i)) = (-1)^it$.  It suffices to show that $g'$ is smooth for each such coordinate system.

Because $\partial_n = N$, $g(\partial_n, \partial_n) = 1$, and therefore
\begin{align}
  \frac{\partial g_{in}}{\partial x^n} = g(\nabla_{\partial_n}\partial_i, \partial_n) = \frac{1}{2}\frac{\partial g_{nn}}{\partial x^i} = 0.
\end{align}
Thus $g_{in}$ is constant on $U$, which implies that $g_{in} = 0$ on $U$.
We shall use $i, j$ to range in $1, \dots, n - 1$.  We can see, therefore, that $g'$ is smooth if and only if $\frac{\partial g_{ij}}{\partial x^n} = 0$ for $n$ odd.  

\begin{lem}\label{lem:lemmahori}
  Let $(M, g)$ be a Riemannian manifold with boundary.  Then the continuous Riemannian metric $g'$ is smooth if and only if $\nabla^{2k + 1}_{\partial_n}\partial_i = 0$ for all $i$ in $1, \dots, n - 1$ and all $k \geq 0$.
\end{lem}
\begin{proof}
  Note that
  \begin{align}
    g(\nabla_{\partial_n}\partial_i, \partial_j) = \frac{1}{2}\left(\frac{\partial g_{ij}}{\partial x^n} + \frac{\partial g_{jn}}{\partial x^i} - \frac{\partial g_{in}}{\partial x^j}\right) = \frac{1}{2}\frac{\partial g_{ij}}{\partial x^n}
  \end{align}
  Furthermore, we have
  \begin{align}
    \frac{1}{2}\left(\partial_n\right)^{2k + 1} g_{ij}= (\partial_n)^{2k}g(\nabla_{\partial_n}\partial_i, \partial_j) = \sum_{l = 0}^{2k} \binom{2k}{l} g((\nabla_{\partial_n})^{l + 1}\partial_i, (\nabla_{\partial_n})^{2k - l}\partial_j).
  \end{align}
  The result now directly follows by induction on $k$.
\end{proof}

Lemma \ref{lem:lemmahori} can be reformulated in more geometric terms.  Let $R$ be the Riemann curvature and define $\widehat{R}(X) := R(\partial_n, X)\partial_n$.  We can now state Mori's theorem.

\begin{thm}
  Let $(M, g)$ be a Riemannian manifold with boundary.  Then the continuous Riemannian metric $g'$ is smooth if and only if $\partial M$ is a totally geodesic submanifold and $\nabla_{\partial_n}^{2k + 1}\widehat{R} = 0$ for all $k \geq 1$.
\end{thm}
\begin{proof}
  Note that $\widehat{R}(\partial_i) = \nabla_{\partial_n}\nabla_{\partial_i}\partial_n = \nabla_{\partial_n}^2\partial_i$.  Therefore, for $k \geq 1$
  \begin{align}
    \nabla^{2k + 1}_{\partial_n}\partial_i = \nabla_{\partial_n}^{2k - 1}(\widehat{R}(\partial_i)) = \sum_{l = 0}^{2k - 1} \binom{2k - 1}{l}(\nabla_{\partial_n}^l \widehat{R})(\nabla_{\partial_n}^{2k - 1 - l}\partial_i)
  \end{align}
  The result now follows from the above formula and the fact that $\partial M$ is totally geodesic if and only if $\nabla_{\partial_n}\partial_i = 0$ for all $i$ in $1, \dots, n - 1$.
\end{proof}

\begin{defin}\label{def:doublable}
  We call a Riemannian manifold satisfying the conditions of the above theorem \emph{doublable}.
\end{defin}

\begin{cor}
  Let $M$ be a Riemannian manifold with boundary with parallel curvature tensor and totally geodesic boundary.  Then $M$ is doublable.
\end{cor}
Note that a manifold without boundary with a parallel Riemann curvature tensor is the same as a locally symmetric space (see \cite{Wolf_2010}).  

\end{appendices}

\section*{Acknowledgments}
\label{sec:acknowled}
Thanks to Kevin Costello and Eugene Rabinovich for independently suggesting the possibility that an additional Taylor expansion normal to the boundary would produce local counterterms.  Additional thanks to Kevin Costello for suggesting a Taylor expansion of the geodesic distance squared as a function of two variables as a possible way of going beyond the parallel curvature assumption. 

\printbibliography[title={Bibliography}]

\end{document}